\documentclass{article}
\usepackage[utf8]{inputenc}
\usepackage[T1]{fontenc}
\usepackage{amsmath}
\usepackage{dsfont}
\usepackage{fancyhdr}
\usepackage[table,xcdraw]{xcolor}
\usepackage{graphicx}
\usepackage{multirow}
\usepackage{endnotes}
\usepackage{float}
\usepackage[nottoc]{tocbibind}
\usepackage{vmargin}
\usepackage{amssymb}
\usepackage{subfigure}
\usepackage{import}
\usepackage{wrapfig}
\usepackage{amsthm}
\usepackage{rotating}
\usepackage[title]{appendix}
\usepackage{listings}

\usepackage{soul}

\usepackage{stmaryrd}
\usepackage{tikz}
\usetikzlibrary{arrows.meta, positioning}

\usepackage{mleftright}
\mleftright
\usetikzlibrary{matrix}
\usepackage{tikz}
\usetikzlibrary{shapes,arrows,chains, decorations.pathmorphing, decorations.pathreplacing}
\tikzset{quantum/.style={decorate, decoration=snake}}

\newcommand{\ket}[1]{|#1\rangle}

\newcommand{\ketbra}[2]{|#1\rangle\langle#2|}

\newcommand{\abs}[1]{\lvert #1\rvert}

\newcommand{\norm}[1]{\| #1\|}

\newcommand{\QPVBB}{$\mathrm{QPV}_{\mathrm{BB84}}$}
\newcommand{\QPVBBparallel}[1]{$\mathrm{QPV}_{\mathrm{BB84}}^{\times #1}$}

\newcommand{\QPVBBf}{$\mathrm{QPV}_{\mathrm{BB84}}^{f}$}

\newcommand{\QPVBBfparallel}{$\mathrm{QPV}_{\mathrm{BB84}}^{f: n\rightarrow m}$}

\newcommand{\routing}{$\mathrm{QPV}_{\mathrm{rout}}$}
\newcommand{\routingf}{$\mathrm{QPV}_{\mathrm{rout}}^{f}$}
\newcommand{\routingparallel}[1]{$\mathrm{QPV}_{\mathrm{rout}}^{\times #1}$}
\newcommand{\routingfparallel}{$\mathrm{QPV}_{\mathrm{rout}}^{f: n\rightarrow m}$}

\newcommand{\tr}[1]{\mathrm{Tr}\left[#1\right]}
\newcommand{\ptr}[2]{\mathrm{Tr}_{#1}\left[#2\right]}

\usetikzlibrary{decorations.pathreplacing,calligraphy}
\usetikzlibrary{matrix}
\usepackage{tikz}
\usetikzlibrary{shapes,arrows,chains, decorations.pathmorphing, decorations.pathreplacing}
\tikzset{quantum/.style={decorate, decoration=snake}}

\newcommand{\diagdots}[3][-25]{%
  \rotatebox{#1}{\makebox[0pt]{\makebox[#2]{\xleaders\hbox{$\cdot$\hskip#3}\hfill\kern0pt}}}%
}

\usepackage{slashed}
\usepackage{enumitem}

\usepackage{authblk}

\usepackage{tcolorbox} % For rounded corners and box customization
\usepackage{array}     % For table control
\usepackage{colortbl}  % For table colors
\usepackage{lipsum}

\usepackage{MnSymbol}

\definecolor{secondaryColor}{HTML}{5869bc}

\usepackage{hyperref}
\hypersetup{colorlinks, allcolors=secondaryColor}
\usepackage{url}
\usepackage{cleveref}
\theoremstyle{plain}
\newtheorem{theorem}{Theorem}[section]

\newtheorem{remark}[theorem]{Remark}
\newtheorem{prop}[theorem]{Proposition}
\newtheorem{definition}[theorem]{Definition}

\newtheorem{lemma}[theorem]{Lemma}
\newtheorem{corollary}[theorem]{Corollary}

\usepackage[normalem]{ulem}

\crefname{figure}{Fig.}{Fig.}
\crefname{theorem}{Theorem}{Theorems}
\crefname{prop}{Proposition}{Propositions}
\crefname{observation}{Observation}{Observations}
\crefname{definition}{Definition}{Definitions}
\crefname{ex}{Example}{Examples}
\crefname{lemma}{Lemma}{Lemmas}
\crefname{corollary}{Corollary}{Corollaries}
\crefname{result}{Result}{Results}
\crefname{attack}{Attack}{Atacks}
\crefname{table}{Table}{Table}
\crefname{section}{Section}{Section}

\usepackage[table,xcdraw]{xcolor} % Required for table coloring
\usepackage{amssymb} % Required for checkmark and cross symbols
\usepackage{pifont} % For cross mark (X)

\begin{document}
\title{Quantum position verification in one shot: parallel repetition of the $f$-BB84 and $f$-routing protocols}

\author[1,2]{Lloren\c{c} Escol\`a-Farr\`as}
\author[1,2]{Florian Speelman}
\newcommand{\lle}[1]{{\color{blue}#1}}
\newcommand{\fs}[1]{{\textcolor{red}{[Florian: #1]}}}

\affil[1]{QuSoft, CWI Amsterdam, Science Park 123, 1098 XG Amsterdam, The Netherlands}
\affil[2]{Multiscale Networked Systems (MNS),  Informatics Institute, University of Amsterdam, Science Park 904,  1098 XH Amsterdam, The Netherlands }

\renewcommand\Affilfont{\itshape\small}
\maketitle
\vspace{-1em}
\begin{abstract}
Quantum position verification (QPV) aims to verify an untrusted prover's location by timing communication with them. To reduce uncertainty, it is desirable for this verification to occur in a single round. However, previous protocols achieving one-round secure QPV had critical drawbacks: attackers pre-sharing an EPR pair per qubit could perfectly break them, and their security depended on quantum information traveling at the speed of light in vacuum, a major experimental challenge in quantum networks. In this work, we prove that a single round of interaction suffices for secure position verification while overcoming these limitations. We show that security for a one-round protocol can rely only on the size of the classical information rather than quantum resources, making implementation more feasible, even with a qubit error tolerance of up to $3.6\%$, which is experimentally achievable with current technology --- and showing that the timing constraints have to apply only to classical communication. In short, we establish parallel repetition of the $f$-BB84 and $f$-routing QPV protocols. As a consequence of our techniques, we also demonstrate an order-of-magnitude improvement in the error tolerance for the sequential repetition version of these protocols, compared to the previous bounds of \emph{Nature Physics  18, 623–626 (2022)}. 

\end{abstract}
\tableofcontents

\newpage
\section{Introduction}

Position verification (PV) is a cryptographic primitive that consists of securely determining the position of an untrusted party $P$, forming a central part of the field of \emph{position-based cryptography}. Classical PV has been shown to be insecure~\cite{OriginalPositionBasedCryptChandran2009}, even under computational assumptions, due to a general attack based on copying information. However, quantum mechanics circumvents this attack via the no-cloning theorem~\cite{Wootters1982NoCloning}, which prevents perfect copying of unknown quantum states. This insight led to the first PV protocols using quantum information ~\cite{PatentKentANdOthers,OriginalQPV_Kent2011,Malaney_QLP} ---known in the literature as quantum position verification (QPV) protocols. The general setting for a one-dimensional\footnote{This is the case that captured most of the attention in literature. Ideas generalize to multiple dimensions, however, some care has to be taken when introducing nonnegligible timing uncertainty in the 3D case.} QPV protocol is described by two trusted verifiers $V_0$ and $V_1$ located in a straight line at the left and at the right of $P$, respectively, who is supposed to be at the position $pos$. The two verifiers are assumed to have synchronized clocks and send quantum or classical messages to $P$ at the speed of light. In a negligible time, $P$ has to pass a challenge using the information that she received and answer back to the verifiers at the speed of light as well. The verifiers \emph{accept} the location if they received correct answers according to the time that the speed of light would take to reach $P$ and return, otherwise, they \emph{reject}. 

Despite the hope for unconditional security, general attacks that apply to all QPV protocols exist~\cite{Buhrman_2014,Beigi_2011}. In an attack, two adversaries, Alice and Bob, are placed between $V_0$ and $P$ and between $V_1$ and $P$, respectively, and they act as follows: (i) they intercept the messages coming from their respective closest verifier, (ii) and, due to relativistic constraints, they are allowed a single round of simultaneous communication before (iii) responding to the verifiers. The best-known general attack~\cite{Beigi_2011} requires that Alice and Bob pre-share an impractically large---exponential---amount of entanglement prior to the execution of the protocol, i.e.\ before (i). This impracticality has sustained interest in showing security in different attack models in the plain model~\cite{PatentKentANdOthers,OriginalQPV_Kent2011,Lau_2011,https://doi.org/10.48550/arxiv.1504.07171,Chakraborty_2015,speelman2016instantaneous,dolev2019constraining,dolev2022non,gonzales2019bounds,cree2022code,bluhm2022single,gao2016quantum,escolàfarràs2024quantumcloninggameapplications}, as well as with extra assumptions such as the random oracle model \cite{Unruh_2014_QPV_random_oracle} or computational assumptions using a quantum computer~\cite{liu2021beating}.  Security proofs in these models have been shown by either (i) bounding the probability of a successful attack by a constant and amplifying security through sequential repetition over time, or (ii) directly showing that the attack success probability is exponentially small, corresponding to parallel repetition. These upper bounds are referred to as the protocol’s \emph{soundness}. Since QPV is based on timing constraints, parallel repetition implies that the verifiers either \emph{accept} or \emph{reject} the location in a single execution, which is crucial to reduce the uncertainty of the location that is to be verified, as opposed to sequential repetition, where timing constraints accumulate over multiple rounds executed one after the other. Moreover, a single interaction is necessary in order to verify the location of a non-static prover. However, previous parallel repetition results in the literature for QPV required the quantum information to travel at the speed of light in vacuum, which is experimentally challenging\footnote{Whereas the transmission of classical information without loss at the speed of light is technologically feasible, e.g.\ via radio waves, the quantum counterpart faces obstacles. Most QPV protocols require quantum information to be transmitted at the speed of light in vacuum, but for practical applications this is often unattainable, e.g.~the speed of light in optical fibers is significantly lower than in vacuum, or if one wants to use a quantum network, it would be desirable that the infrastructure can be used even if the verifiers and $P$ are not connected in a straight line.}, and remained insecure if attackers used one EPR pair per qubit used in the protocols. In order to implement QPV experimentally, it is essential to eliminate these limitations. In this paper, we bridge this gap.

The core of our work is based on extensions of the BB84 (\QPVBB) and routing (\routing) protocols~\cite{PatentKentANdOthers,OriginalQPV_Kent2011}.  Variants of them have taken a central role in the QPV literature~\cite{Buhrman_2014,TomamichelMonogamyGame2013,Unr14_QuantumPositionVerification,bluhm2022single,Escol_Farr_s_2023,asadi2024lineargateboundsnatural,cree2022code,Allerstorfer_2024,escolàfarràs2024quantumcloninggameapplications}, with many results established for one also applying to the other. The depth of this correspondence remains yet to be explored. In this work, we establish results for an extension of the former, and then demonstrate that they extend to the latter. In the \QPVBB~protocol, $V_0$ and $V_1$ send a BB84 state and a classical bit $z\in\{0,1\}$ to the prover $P$, respectively, then, the prover has to measure the qubit in either the computational ($z=0$) or the Hadamard ($z=1$) basis, and broadcast the outcome to both verifiers, see \cref{fig:parallel_BB84} for a schematic representation of a generalization of the protocol. \QPVBB~was proven to be secure~\cite{Buhrman_2014} in the no pre-shared entanglement  (No-PE) model---where attackers do not pre-share any entanglement prior to the execution of the protocol--- showing constant soundness for a single round, and exponentially decaying soundness when the protocol is executed $m$ times in parallel, \QPVBBparallel{m} \cite{TomamichelMonogamyGame2013}. However, it suffices for Alice and Bob to pre-share a single EPR pair per qubit sent by $V_0$ to perfectly break this protocol~\cite{OriginalQPV_Kent2011}. The latter issue, without parallel repetition, i.e.\ for $m=1$, was bypassed in \cite{Buhrman_2013,bluhm2022single}~by splitting the classical bit $z$ into $n$-bit strings $x,y\in\{0,1\}^n$, sent from $V_0$ and $V_1$, respectively, so that a boolean function $f:\{0,1\}^n\times\{0,1\}^n\rightarrow\{0,1\}$ determines $z$, i.e.\ $z=f(x,y)$. We denote this extension by \QPVBBf. The authors~\cite{bluhm2022single} showed that the protocol has a soundness of at most $0.98$, provided that attackers pre-share a number of qubits linear in $n$---the Bounded-Entanglement (BE($n$)) model. This extension requires any attackers to share an amount of entanglement that grows with the classical information, making it an appealing candidate to aim towards implementation.  Then, in order to either \emph{accept} or \emph{reject}, the verifiers execute \QPVBBf~ sequentially $m$ times.

\subsection{Results}
In this paper, we study \QPVBBf when executed $m$ times in parallel, denoted by \QPVBBfparallel, where the classical information $z\in\{0,1\}^m$ is determined by a function $f:\{0,1\}^n\times\{0,1\}^n\rightarrow\{0,1\}^m$. Unruh~\cite{Unruh_2014_QPV_random_oracle} showed the security of this protocol in the random oracle model, assuming the function $f$ is a hash function modeled as a quantum random oracle (and quantum information traveling at the speed of light in vacuum). Here, we show that this protocol exhibits exponentially decaying soundness in $m$ in the plain model, provided that the number of pre-shared qubits by attackers scales linearly with the classical information $n$. Notably, this implies that security is fundamentally tied to the classical information rather than the quantum resources. Moreover, only the classical information is required to travel at the speed of light whereas the quantum counterpart can be arbitrarily slow. We thus show that a single round of interaction with the prover suffices for secure position verification while overcoming the above-mentioned limitations, preserving exponentially decaying soundness while tolerating an error\footnote{ Because of experimental imperfections, we also study a version of the protocol where the prover only has to answer correctly on a fraction of the parallel rounds.} up to $3.6\%$, which is currently implementable in a laboratory.

As a consequence of our analysis, we are also able to improve the particular case of $m=1$ to show soundness of $0.8539$. This is essentially tight, since it closely matches the best known attack (which does not use any entanglement), which has success probability  $\frac{1}{2}+\frac{1}{2\sqrt{2}}=0.85355...$, and this result constitutes an improvement of an order of magnitude with respect to the 0.98 soundness shown in \cite{bluhm2022single}. Therefore, our new bounds are useful even when only considering sequential repetition of \QPVBBf. See \cref{table:comparison} for a summary of the previously known results of \QPVBB~and its variants together with the new results presented in this work. 

We show similar results for parallel repetition for the analogous extension of the routing protocol, where the prover also receives a BB84 state and a bit $z$ and the task is to send the qubit to the verifier $V_z$. We also show an order-of-magnitude improvement for the case $m=1$, providing essentially a tight result matching the best known attack \cite{escolàfarràs2024quantumcloninggameapplications}, which uses no pre-shared entanglement. However, this protocol has the drawback that quantum information sent by the prover is required to travel at the speed of light in vacuum, nevertheless, it is also an appealing candidate for free-space quantum position verification, since the hardware of the prover could hypothetically be as simple as an adjustable mirror or an optical switch.

Our main results are informally stated as follows:

\begin{theorem} (Informal) If attackers pre-share a number of qubits which is linear in the classical information $n$, then, the \QPVBBfparallel  protocol, in the error-free case, has exponentially small soundness, behaving as 
\begin{equation*}
     (0.853909\ldots)^m.
\end{equation*}
Moreover, the protocol still has exponentially small soundness even with a qubit error up to $3.6\%$. Furthermore, it remains secure even if the quantum information sent by $V_0$ travels arbitrarily slow.     
\end{theorem}

\begin{theorem} (Informal) If attackers pre-share a number of qubits which is linear in the classical information $n$, then, the \routingfparallel protocol has  exponentially (in $m$) small soundness, behaving as 
\begin{equation*}
    (0.750436\ldots)^m.
\end{equation*}   
Moreover, the protocol still has exponentially small soundness even with a qubit error up to $3.0\%$.
\end{theorem}

\begin{table}[ht]
\centering
\begin{tabular}{c|c|ccc|c}
                               & No-PE model & \multicolumn{3}{c|}{BE model}   & Slow quantum \\ \hline
Protocol                       & Secure                                  & Sec.\ vs EPR & Soundness & Error  & Secure                               \\ \hline
\QPVBB          & $\checkmark$                             & \ding{55}           & --        & --          & \ding{55}                             \\ 
\QPVBBparallel{m}  & $\checkmark$                           & \ding{55}           & $O(2^{-m})$ in BE($m$) & $3.7\%$         & \ding{55}                              \\ 
\QPVBBf         & $\checkmark$                            & $\checkmark$       & $O(1)$\hspace{4mm} in BE($n$)  & \cellcolor{lightgray}$2\%\rightarrow 14.6\%$           & $\checkmark$                            \\ 
\QPVBBfparallel & $\checkmark$                            & \cellcolor{lightgray}$\checkmark$         & \cellcolor{lightgray}$O(2^{-m})$ in BE($n$) &\cellcolor{lightgray} $3.6\%$         &\cellcolor{lightgray} $\checkmark$                            \\ 
\end{tabular}
\caption{Summary of results about \QPVBB~and its variants. We highlight in gray background the cells with the new results presented in this work. `Sec.\ vs EPR' means that the protocol is secure if attackers pre-share one EPR pair per qubit in the protocol. BE($m$) and BE($n$) denote that the security parameter in the Bounded-Entanglement model is the quantum information $m$ and the classical information $n$, respectively. The soundness column denotes the soundness per round,  \QPVBBf~achieves exponential soundness by sequential repetition. The column corresponding to `Slow quantum' answers whether the protocol is secure even if the quantum information in an execution of the protocol travels arbitrarily slow.}
\label{table:comparison}
\end{table}

\section{Preliminaries}
For $n\in\mathbb N$, we will denote $[n]:=\{0,\ldots,n-1\}$. Let $\mathcal{H}$, $\mathcal{H'}$ be finite-dimensional Hilbert spaces, we denote by $\mathcal{B}(\mathcal{H},\mathcal{H'})$ the set of bounded operators from $\mathcal{H}$ to $\mathcal{H'}$ and $\mathcal{B}(\mathcal{H})=\mathcal{B}(\mathcal{H},\mathcal{H})$. For $A,B\in\mathcal{B}(\mathcal H)$, we denote $A\succeq B$ if $A-B$ is positive semidefinite. Denote by $\mathcal{S}(\mathcal{H})$ the set of quantum states on $\mathcal{H}$,~i.e.\ $\mathcal{S}(\mathcal{H})=\{\rho\in\mathcal{B}(\mathcal{H})\mid \rho\geq0, \tr{\rho}=1)\}$. A pure state will be denoted by a ket $\ket{\psi}\in\mathcal{H}$. An EPR pair is the state $\ket{\Phi^+}=\frac{1}{\sqrt{2}}(\ket{00}+\ket{11}).$ We will refer to basis 0 and 1 to denote the computational and Hadamard basis, respectively. The Hadamard transformation will be denoted by $H$.   For bit strings $x,a\in\{0,1\}^n$ we denote 
\begin{equation}
    \ketbra{a^x}{a^x}:=\ketbra{a_1^{x_1}}{a_1^{x_1}}\otimes\ldots\otimes\ketbra{a_{n}^{x_{n}}}{a_{n}^{x_{n}}},
\end{equation}
where $\ketbra{a_i^{x_i}}{a_i^{x_i}}=H^{x_i}\ketbra{a_i}{a_i}H^{x_i}$. 
The Hamming weight $w_H$ of a bit string $x\in\{0,1\}^n$ is the number of 1's in $x$, i.e. $w_H(x):=\abs{\{i\in[n]\mid x_i=1\}}$. The Hamming distance $d_{H}$ between two bit strings $x,y\in\{0,1\}^n$ is the number of positions at which they differ, i.e. $d_{H}(x,y):=\abs{\{i\in[n]\mid x_i\neq y_i\}}.$

We will use $\log$ for the logarithm in basis 2. The function $h(p):=-p\log p-(1-p)\log(1-p)$ denotes the binary entropy.  For $s\in[0,1]$, and $m\in\mathbb N$, we will use the notation
\begin{equation}
    \mathcal P_{\leq s}(\{0,1\}^m)=\{S\subseteq\{0,1\}^m\mid \abs{S}\leq 2^{sm}\},
\end{equation}
for the set of subsets of $\{0,1\}^m$ of size at most $2^{sm}$. The following values will appear often, and we will use the following shorthand notation
\begin{equation}\label{eq:def_lambda_gamma}
    \lambda_\gamma:=2^{h(\gamma)}\left(\frac{1}{2}+\frac{1}{2\sqrt{2}}\right),
    \end{equation}
    and
    \begin{equation}\label{eq:def_mu_gamma}
    \mu_\gamma:=2^{\gamma+h(\gamma)}\left(\frac{1}{2}+\frac{1}{2\sqrt{2}}\right),
\end{equation}
and we will use $\lambda_0=\mu_0=\frac{1}{2}+\frac{1}{2\sqrt{2}}$. 

\section{Parallel repetition of \QPVBBf}\label{sec:parallel_f-BB84}
In this section, we study the $m$-fold parallel repetition of \QPVBBf, which we denote by \QPVBBfparallel. We will describe the protocol, its general attack, and we will prove that the protocol exhibits exponentially small soundness in the quantum information $m$ provided that the attacker's amount of pre-shared entanglement is linearly bounded by the size of the classical information $n$, i.e. in the BE($n$) model. 

\begin{definition} \label{def qpv bb84 f} \emph{(\QPVBBfparallel~protocol)}.
Let $n,m\in\mathbb{N}$, and $f:\{0,1\}^n \times \{0,1\}^n \to \{0,1\}^m$, and consider an error parameter $\gamma\in[0,\frac{1}{2})$. The \QPVBBfparallel~protocol is described as follows:
\begin{enumerate}
    \item The verifiers $V_0$ and $V_1$ secretly agree on bit strings $x,y\in\{0,1\}^n$ and $a\in\{0,1\}^m$, chosen uniformly at random. Then,  $V_0$ prepares the $m$-qubit state $H^{f(x,y)_1}\ket{a_1}\otimes \dots \otimes H^{f(x,y)_m}\ket{a_m}=:H^{f(x,y)}\ket{a}$. 
    \item Verifier $V_0$ sends $H^{f(x,y)}\ket{a}$ and $x\in\{0,1\}^n$ to $P$, and $V_1$ sends $y\in\{0,1\}^n$ to $P$ so that all the information arrives at $pos$ simultaneously. The classical information is required to travel at the speed of light, whereas the quantum information can be sent arbitrarily slow.
    \item Immediately, $P$ measures each qubit $H^{f(x,y)_i}\ket{a_i}$ in the basis $f(x,y)_i=:z_i$ for all $i\in[m]$ ($z:=f(x,y)$), and broadcasts her outcome ${v \in \{0, 1 \}^m}$ to $V_0$ and $V_1$. 
    \item  The verifiers \emph{accept} if $d_H(a,v)\leq \gamma m$ (consistency with the error), and $v$ arrives at the time consistent with $pos$. If either the answers do not arrive on time or are different, the verifiers \emph{reject}.  
\end{enumerate}
\end{definition}
See Fig.~\ref{fig:parallel_BB84} for a schematic representation of the \QPVBBfparallel~protocol. The \QPVBB~ and \QPVBBparallel{m}~protocols are recovered 
if the only classical information that is sent from the verifiers is $y\in\{0,1\}$ and $y\in\{0,1\}^m$, respectively (and $z=y$), and \QPVBBf~is recovered by setting $m=1$. 

\begin{figure}[h]
    \centering
    \scalebox{0.9}{
    \begin{tikzpicture}[node distance=3cm, auto]
    \node (A) {$V_0$};
    \node [left=1cm of A] {};
    \node [right=of A] (B) {$P$};
    \node [right=of B] (C) {$V_1$};
    \node [right=1cm of C] {};
    \node [below=of A] (D) {};
    \node [below=of B] (E) {};

    \node [above=-0.1cm of A] (N) {};
    \node [right=0.8cm of N] {$\diagdots[120]{2.5em}{0.1em}$};
    
    \node [right=0.8cm of A] (M) {};
    \node [left=1.5cm of C] (M2) {};
    
    \node [below=of C] (F) {};
    \node [below=of D] (G) {};%$V_0$
    \node [below=of E] (H) {};
    \node [below=of F] (I) {};%$V_1$
    \node [left= 6cm of E] (J) {};
    \node [below= 3cm of J] (K) {};
    \node [above= 3cm of J] (L) {};
    \node [right=0.1 of E](P fxy){$f(x,y)=z$};

    \draw [->, transform canvas={xshift=0pt, yshift=0pt}, quantum] (M) -- (E) node[midway] (x) {} ;
    \draw [->] (A) -- (E);
    \draw [->] (C) -- (E);
    \draw [][->] (E) -- (I) node[midway] (q) {$v\in\{0,1\}^m$}; %line width=0.4mm, dotted
    \draw [][->] (E) -- (G); %line width=0.4mm, dotted

    \draw [->] (L) -- (K) node[midway] {time};

    \node[below=0.5cm of B] {$\otimes_{i=1}^mH^{f(x,y)_i}\ket{a_i}$};
    
    \node[left=1.4cm of x, transform canvas={xshift=+ 2pt, yshift = +2 pt}] {$x\in\{0,1\}^n$};
    \node[right = 2.9cm of x, transform canvas={xshift=+ 2pt, yshift = +2 pt}] {$y \in \{0,1\}^n$};
    \node[left = 3.5cm of q] {$v\in\{0,1\}^m$};

    \node [above=0.5cm of A] (posV00) {};
    \node [left=1cm of posV00] (posV0) {};
    \node [above=0.5cm of C] (posV11) {};
    \node [right=1cm of posV11] (posV1) {};
    \draw [->] (posV0) -- (posV1) node[midway] {position};

    \node [right=0.8cm of C] (VP0) {};
    \node [right=4.55cm of E] (VPP) {};
    \node [right=0.8cm of I] (PV) {};
    \end{tikzpicture}
    }
\caption{Schematic representation of the \QPVBBfparallel~protocol. Undulated lines represent quantum information, whereas straight lines represent classical information. The slowly travelling quantum system $\otimes_{i=1}^mH^{f(x,y)_i}\ket{a_i}$ originated from $V_0$ in the past.}
\label{fig:parallel_BB84}
\end{figure}

For the security analysis, we will consider the purified version of \QPVBBfparallel, which is equivalent to it. The difference relies on, instead of $V_0$ sending BB84 states, $V_0$ prepares $m$  EPR pairs $\ket{\Phi^+}_{V_0^1Q_1}\otimes\dots\otimes\ket{\Phi^+}_{V_0^mQ_m}$ and sends the registers $Q_1\ldots Q_m$ to the prover. In a later moment, the verifier $V_0$ performs the measurement  $\{H^{f(x,y)}\ketbra{a}{a}_VH^{f(x,y)}\}_{a\in\{0,1\}^m}$ in his local registers $V_0^1\ldots V_0^m=:V$. In this way, the verifiers delay the choice of basis in which the $m$ qubits are encoded, which, in contrast to the above prepare-and-measure version, will make any attack independent of the state sent by $V_0$. 

The most general attack on a 1-dimensional QPV protocol consists on placing an adversary between $V_0$ and the prover, Alice, and another adversary between the prover and $V_1$, Bob. In order to attack \QPVBBfparallel, 
\begin{enumerate} \item Alice intercepts the $m$ qubit state $Q_1\ldots Q_m$ and applies an arbitrary quantum operation to it and to a local register that she possess, possibly entangling them. She keeps part of the resulting state, and sends the rest to Bob. Since the qubits $Q_1\ldots Q_m$  can be sent arbitrarily slow by $V_0$ (the verifiers only time the classical information), this happens before Alice and Bob can intercept $x$ and $y$. 

\item Alice intercepts $x$ and Bob intercepts $y$. At this stage, Alice, Bob, and $V_0$ share a quantum state $\ket{\varphi}$, make a partition and let $q$ be the number of qubits that Alice and Bob each hold, recall that $m$ qubits are held by $V_0$ and thus the three parties share a quantum state $\ket{\varphi}$ of $2q+m$ qubits.  Alice and Bob apply a unitary $U_{A_\text{k}A_\text{c}}^{x}$ and $V_{B_\text{k}B_\text{c}}^{y}$ on their local registers $A_\text{k}A_\text{c}=:A$ and $B_\text{k}B_\text{c}=:B$, both of dimension $d=2^{q}$, where k and c denote the registers that will be kept and communicated, respectively. Due to the Stinespring dilation, we consider unitary operations instead of quantum channels. They end up with the quantum state ${\ket{\psi_{xy}}=\mathbb{I}_{V}\otimes U_{A_\text{k}A_\text{c}}^x\otimes V_{B_\text{k}B_\text{c}}^y\ket\varphi}$. Alice sends register $A_\text{c}$ and $x$ to Bob (and keeps register $A_\text{k}$), and Bob sends register $B_\text{c}$ and $y$ to Alice (and keeps register $B_\text{k})$. 

\item Alice and Bob perform POVMs  $\{ A^{xy}_{a}\}_{a\in\{0,1\}^m}$ and $\{ B^{xy}_{b}\}_{b\in\{0,1\}^m}$ on their local registers $A_\text{k}B_\text{c}=:A'$ and $B_\text{k}A_\text{c}=:B'$, and answer their outcomes $a$ and $b$ to their closest verifier, respectively. 
\end{enumerate}

See \cref{fig:attack-parallel_repBB84} for a schematic representation of the general attack to \QPVBBfparallel. The tuple $S=\{\ket\varphi,U^x,V^y,\{A^{xy}_a\}_a,\{B^{xy}_b\}_b\}_{x,y}$ will be called a $q$-qubit strategy for \QPVBBfparallel. Then, the probability that Alice and Bob perform a successful attack up to error $\gamma$, provided the strategy $S$, which we denote by $\omega_S$, is given by

\begin{equation}\label{eq:w_S}
\begin{split}
       \omega_{S}(\text{\QPVBBfparallel})=\frac{1}{2^{2n} }\sum_{x,y,a}\tr{ \left(H^{f(x,y)}\ketbra{a}{a}_VH^{f(x,y)} \otimes \sum_{a':d_H(a,a')\leq \gamma m}A^{xy}_{a'}\otimes B^{xy}_{a'}\right)\ketbra{\psi_{xy}}{\psi_{xy}}_{VA'B'}}.
\end{split}
\end{equation}
The optimal attack probability is given by
\begin{equation}
    \omega^*(\text{\QPVBBfparallel})=\sup_{S}\omega_S(\text{\QPVBBfparallel}),
\end{equation}
where the supremum is taking over all possible strategies $S$. As mentioned above, the existence of a generic attack for all QPV protocols \cite{Beigi_2011,Buhrman_2014} implies that $\omega^*(\text{\QPVBBfparallel})$ can be made arbitrarily close to 1. However, the best known attack requires an exponential amount of pre-shared entanglement. Therefore, we will study the optimal winning probability under restricted strategies 
$S$, specifically imposing a constraint on the number of pre-shared qubits $q$ that Alice and Bob hold in step 2 of the general attack. 
Throughout this section, we adopt the following notation to enhance readability:
\begin{enumerate}
    \item we  omit (\text{\QPVBBfparallel}) in $\omega_{S}(\text{\QPVBBfparallel})$, and its variants (see below),
    \item we define 
    \begin{equation}
    M^{f(x,y)}_a:=H^{f(x,y)}\ketbra{a}{a}_VH^{f(x,y)},
\end{equation}
for the measurement that $V_0$ performs, and
\item given a strategy $S=\{\ket\varphi,U^x,V^y,\{A^{xy}_a\}_a,\{B^{xy}_b\}_b\}_{x,y}$, we introduce
\begin{equation}
       \Pi_{AB}^{xy}:=\sum_a\left(M^{f(x,y)}_a \otimes \sum_{a':d_H(a,a')\leq \gamma m}A^{xy}_{a'}\otimes B^{xy}_{a'}\right),
   \end{equation}
\end{enumerate}
in this way, we have
\begin{equation}
    \omega_S=\frac{1}{2^{2n}}\sum_{x,y}\tr{\Pi^{xy}_{AB}\ketbra{\psi_{xy}}{\psi_{xy}}}.
\end{equation}

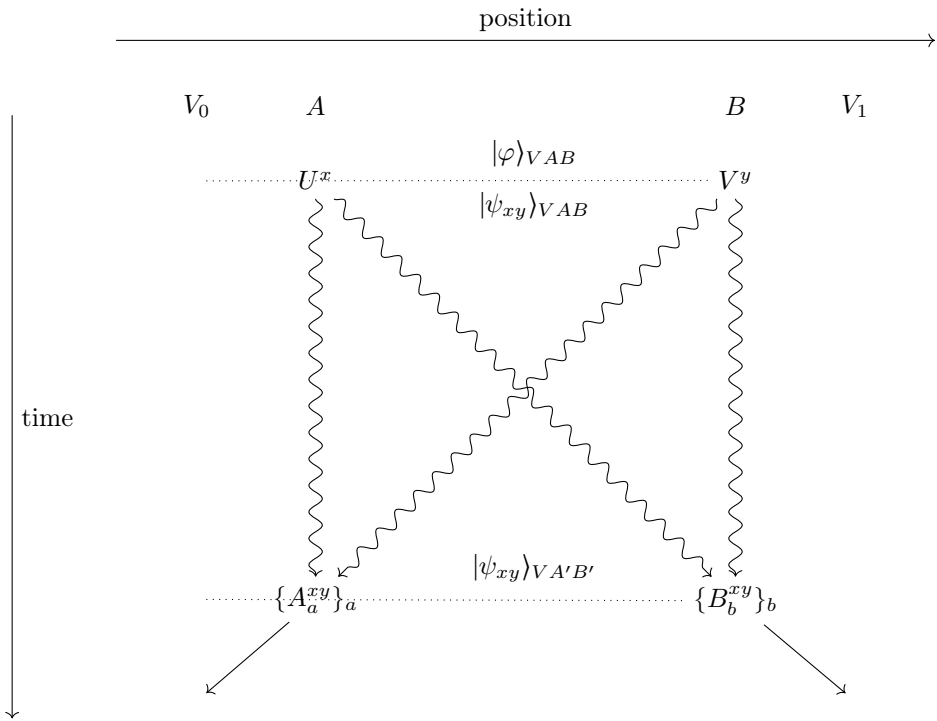
\begin{figure}[h]
    \centering
    \begin{tikzpicture}[node distance=3cm, auto]
    \node (V0) {$V_0$};   
    \node [left=2cm of V0](t0){};
    \node [below=8cm of t0](t1){};
    \draw [->] (t0) -- (t1) node[midway] {time};

    \node [right=1cm of V0] (A) {$A$};
    \node [right= 5cm  of A] (B) {$B$};
    \node [right= 2.5cm  of A] (P) {};
    \node [right=1cm of B] (V1) {$V_1$};
    \node [below =0.6cm of V0](below_V0){};
    \node [below =0.5cm of V1](below_V1){};

    \node [below =6.15cm of V0](V_ABans){};

    \node [below =7.5cm of V0](V0_ans){};
    \node [below =7.5cm of V1](V1_ans){};
    
    \node [below =0.2cm of P](middle){$\ket\varphi_{VAB}$};
    \node [above =0.25cm of middle](middle0){};

    \node [below =0.1cm of middle](middle0){$\ket{\psi_{xy}}_{VAB}$};
    \node [below =0.5cm of A](A_intercepts_Q){};
    \node [below =0.5cm of B](B_intercepts_Q){};
    \node [below =0.5cm of A](A_intercepts_x){$U^x$};
    \node [below =0.5cm of B](B_intercepts_y){$V^y$};

    \node [above=0.5cm of A] (posV00) {};
    \node [left=2.5cm of posV00] (posV0) {};
    \node [above=0.5cm of B] (posV11) {};
    \node [right=2.5cm of posV11] (posV1) {};
    \draw [->] (posV0) -- (posV1) node[midway] {position};

    \node [below =5cm of A](A_commits){};
    \node [below =5cm of B](B_commits){};
    \node [below =6cm of A](A_answers){$\{A^{xy}_a\}_a$};
    \node [below =0.1cm of A_answers, ](A_POVM){};
    \node [below =6cm of B](B_answers){$\{B^{xy}_b\}_b$};
    \node [below =0.1cm of B_answers](B_POVM){};

    \node [left=0.3cm of A](leftA){};
    \node[below=5.6cm of leftA](cA){};
    \node [left=0.2cm of A](leftA2){};
    \node[below=6.6cm of leftA2](aA){};

    \node [right=0.3cm of B](leftB){};
    \node[below=5.6cm of leftB](cB){};
    \node [right=0.2cm of B](leftB2){};
    \node[below=6.6cm of leftB2](bB){};
    
    \node [below =6.2cm of V0](V0bis){};
    \node [below =6.2cm of V1](V1bis){};
    \node [below =7.2cm of V0](V0bis1){};
    \node [below =7.2cm of V1](V1bis1){};

    \draw [->, transform canvas={xshift=0pt, yshift = 0 pt}, quantum] (A_intercepts_x) -- (B_answers) node[midway] (x) {};

    \draw [->, transform canvas={xshift=0pt, yshift = 0 pt}, quantum] (B_intercepts_y) -- (A_answers) node[midway] (x) {};

    \draw [->, quantum] (A_intercepts_x) -- (A_answers) {};
    \draw [->, quantum] (B_intercepts_y) -- (B_answers) {};
    %\draw [->] (A_commits) -- (cA) {};
    \draw [->] (A_answers) -- (V0_ans) {};
    %\draw [->] (B_commits) -- (cB) {};
    \draw [->] (B_answers) -- (V1_ans) {};

    \draw[dotted][-](below_V0)--(B_intercepts_y){};
    \draw[dotted][-](V_ABans)--(B_answers){};
    \node [right =2.5cm of A_answers](middle2){};
    \node [below =4.9cm of middle](middle3){$\ket{\psi_{xy}}_{VA'B'}$};
\end{tikzpicture}
\caption{Schematic representation of a general attack on \QPVBBfparallel, where straight lines represent classical information, and undulated lines represent quantum information, including $x$ and $y$. Replacing $\{A^{xy}\}_a$ and $\{B^{xy}\}_b$ by $L^{xy}$ and $K^{xy}$, and the srtaight arrows comming out of the attackers by onbdulated lines, representing $A_0'$ and $B_0'$, respectively, this corresponds to a schematic representation of \routingfparallel.}
\label{fig:attack-parallel_repBB84}
\end{figure}

Ideally, Alice and Bob should prepare a `good enough' attack for every $(x,y)\in\{0,1\}^{2n}$, however, we do not have control of what potential attackers might do.  For this reason, we introduce the following concept for attacks that are 'good enough' for a certain set of pairs of $(x,y)$, meaning that for those pairs they have a probability of successfully attacking the protocol which is above a certain threshold $\omega_0$, which defines `good enough'.     

\begin{definition}\label{def:q-beta-strategy} Let $\omega_0,\beta\in(0,1]$. A $q$-qubit strategy $S$ for \QPVBBfparallel~is a $(\omega_0,q,\beta\cdot2^{2n})$-strategy for \QPVBBfparallel~if there exists a set $\mathcal{B}\subseteq\{0,1\}^{2n}$ with $\abs{\mathcal{B}}\geq \beta\cdot2^{2n}$  such that 
\begin{equation}
    \tr{\Pi^{xy}_{AB}\ketbra{\psi_{xy}}{\psi_{xy}}}\geq \omega_0, \text{ }\text{ } \forall(x,y)\in\mathcal{B}.
\end{equation}
\end{definition}

Notice that the choice of the function $f$ will determine the probability distribution of the basis $f(x,y)=z\in\{0,1\}^m$ in which the $m$ qubits have to be measured in the protocol. We denote this probability distribution by $q_f(z)$, which is given by 
\begin{equation}
    q_f(z)=\frac{\abs{\{x,y\mid f(x,y)=z\}}}{2^{2n}}=:\frac{n_z}{2^{2n}},
\end{equation}
where we denote by $n_z$ the number of pairs $(x,y)$ such that $f(x,y)=z$. We say that $f$ \emph{reproduces a uniform distribution over $z\in\{0,1\}^m$} if $q_f(z)=\frac{1}{2^m}\forall z\in\{0,1\}^m$.

In \cite{TomamichelMonogamyGame2013}, the security of the $m$-fold parallel repetition of \QPVBB ~(\QPVBBparallel{m}) was analyzed in the No-PE model, and the authors showed that the protocol has exponentially small (in the quantum information $m$) soundness, provided that the quantum information travels at the speed of light. 

Consider now the \emph{fixed initial-state} (FIS) attack model, which we define as the attack model where step 2. in the general attack is constrained by imposing $\ket{\psi_{xy}}\rightarrow \ket{\psi}$ for all $x,y\in\{0,1\}^n$, i.e. strategies of the form $S_{\text{FIS}}=\{\ket\varphi,U^x=\mathbb{I},V^y=\mathbb{I},\{A^{xy}_a\}_a,\{B^{xy}_b\}_b\}_{x,y}$. Then, the same reduction to a monogamy-of-entanglement game as in \cite{TomamichelMonogamyGame2013} to show security of \QPVBBparallel{m} holds for \QPVBBfparallel. In particular, we have that for all functions  $f$ such that reproduce a uniform distribution on the bases in which the qubits have to be measured, i.e. $q_f(z)=\frac{1}{2^m}$ for all $z\in\{0,1\}^m$, the result in \cite{TomamichelMonogamyGame2013} translates to the following lemma. Not surprisingly, the reduction can be done to strategies $S_{\text{FIS}}$ where $\{A^{xy}_a\}_a$ and $\{B^{xy}_b\}_b$ only depend on $z=f(x,y)$ instead of $x$ and $y$, i.e. $\{A^{z}_a\}_a$ and $\{B^{z}_b\}_b$, see proof of Lemma~\ref{lemma:w_FS}. 

\begin{lemma}\label{lemma:w_FS}(Adapted version of eq. (9) in \cite{TomamichelMonogamyGame2013}). For every function $f$ such that reproduces a uniform distribution over $z\in\{0,1\}^m$, the following holds for \QPVBBfparallel:
\begin{equation}
\omega^*_\text{FIS}:=\sup_{S_\text{FIS}}\omega_{S_\text{FIS}}\leq \left(\lambda_\gamma\right)^m. 
\end{equation} 
\end{lemma}

Recall that $\lambda_\gamma$ is defined in \eqref{eq:def_lambda_gamma}.
\begin{remark}\label{rmk:Breidbart_attack}
    For $\gamma=0$, \cref{lemma:w_FS} is tight, since there exists a strategy \cite{TomamichelMonogamyGame2013} consisting of $V_0$, Alice and Bob preparing sharing the state $\ket\psi_{V}=\otimes_{i=1}^m(\cos\frac{\pi}{8}\ket0_{V_0^i}+\sin\frac{\pi}{8}\ket1_{V_0^i})$, i.e. sending $m$ times to the so-called Breidbart state, and both attackers answering the $m$-bit string $v=0\ldots0$, which reaches the upper bound $(\frac{1}{2}+\frac{1}{2\sqrt 2})^m$. 
\end{remark}

\begin{proof}
From \eqref{eq:w_S}, we have that for $S_{\text{FIS}}=\{\ket\varphi,U^x=\mathbb{I},V^y=\mathbb{I},\{A^{xy}_a\}_a,\{B^{xy}_b\}_b\}_{x,y}$, 
\begin{equation}
    \begin{split}
       \omega_{S_{\text{FIS}}}&=\frac{1}{2^{2n} }\sum_{x,y,a}\tr{ \left(M^{f(x,y)}_a \otimes \sum_{a':d_H(a,a')\leq \gamma m}A^{xy}_{a'}\otimes B^{xy}_{a'}\right)\ketbra{\psi}{\psi}}\\
       &=\sum_{z}\frac{q_f(z)}{n_z}\sum_{a}\sum_{x, y : f(x,y) = z }\tr{ \left(M^{z}_a \otimes \sum_{\substack{ a' : d_H(a, a') \leq \gamma m}}A^{xy}_{a'}\otimes B^{xy}_{a'}\right)\ketbra{\psi}{\psi}}\\
       &\leq\sum_{z}\frac{q_f(z)}{n_z}\sum_a n_z \max_{x,y:f(x,y)=z}\tr{ \left(M^{z}_a \otimes \sum_{ a' : d_H(a, a') \leq \gamma m}A^{xy}_{a'}\otimes B^{xy}_{a'}\right)\ketbra{\psi}{\psi}}.
\end{split}
\end{equation}
Then, denoting by $A^z_{a'}$ and $B^z_{a'}$ the corresponding $A^{xy}_{a'}$ and $ B^{xy}_{a'}$ (recall that these $x$ and $y$ are such that $f(x,y)=z$) that attain the maximum in the last inequality, we have that 
\begin{equation}\label{eqw:S_FS}
    \omega_{S_{\text{FS}}} \leq \frac{1}{2^m}\sum_z \sum_a\tr{\left(M^z_a\otimes \sum_{ a' : d_H(a, a') \leq \gamma m}A^{z}_{a'}\otimes B^{z}_{a'}\right)\ketbra{\psi}{\psi}}.
\end{equation}
In \cite{TomamichelMonogamyGame2013}, it is proven that the right-hand-side of \eqref{eqw:S_FS} is upper bounded by $\left(\lambda_\gamma\right)^m$. 
\end{proof}

A quantity that will be of interest is given by the maximum winning probability whenever the we fix $\ket{\psi}_{VA'B'}$ in a strategy $S_{\text{FIS}}$, we denote this quantity by $\omega_{\psi}^*$, i.e.
\begin{equation}
    \omega_{\psi}^*:=\max_{\{A^{xy}_a\}_a,\{B^{xy}_b\}_b}\frac{1}{2^{2n} }\sum_{x,y,a}\tr{ \left(M^{f(x,y)}_a \otimes \sum_{a':d_H(a,a')\leq \gamma m}A^{xy}_{a'}\otimes B^{xy}_{a'}\right)\ketbra{\psi}{\psi}}.
\end{equation}
As an immediate consequence of Lemma~\ref{lemma:w_FS}, we have:

\begin{corollary}\label{coro:w_psi<=...} For every quantum state $\ket{\psi}_{VA'B'}$, with arbitrary registers $A'$ and $B'$, for every function $f$ such that reproduces a uniform distribution over $z\in\{0,1\}^m$, the following holds for \QPVBBfparallel:
\begin{equation}
    \omega^{*}_{\psi}\leq\left(\lambda_\gamma\right)^m.
\end{equation}
\end{corollary}

\cref{lemma:w_FS} applies for functions $f$ such that reproduce a uniform distribution over $z\in\{0,1\}^m$, however, while not all functions $f$ might be good to use to implement \QPVBBfparallel, e.g. the constant function, only considering uniform distributed values of $z$ restricts the number of functions that we can consider. We will now show that we can still obtain upper bounds for $\omega^*_{FIS}$ for a large class of functions, namely those $f$ that reproduce a distribution over $z$'s that is not very far away from the uniform distribution. This class of functions can be made larger by choosing $n$ larger than $m$, see below.  We will see that if one writes the distribution $q_f(z)$ as the uniform distribution plus a deviation, i.e. 
\begin{equation}
    q_f(z)=\frac{1}{2^m}+\delta_f(z),
\end{equation}
for most of the functions $f$, $\delta_f(z)$ will be small for most of $z\in\{0,1\}^m$. 
In order to analyze the probability distribution over the outputs $z$ induced by a random function $f$, consider the random variable $Q_f(z)=\frac{N_z}{2^m}$ where $N_z$ is the random variable representing the number of times that $z$ appears as an output of $f$. The values that the random variables $Q_f(z)$ and $N_z$ take will be denoted by $q_f(z)$ and $n_z$, respectively. Since $f$ is a random function, $N_z$ follows a binomial distribution 
\begin{equation}
    N_z\sim B\left(2^{2n},\frac{1}{2^m}\right),
\end{equation}
where $2^{2n}$ is the number of trials (possible $x$ and $y$) and $\frac{1}{2^m}$ is the probability of success ({`hitting~$z$'). Then, we have that $\mathbb E_f[N_z]=2^{2n-m}$ and thus, the expected value of $Q_f(z)$ is given by
\begin{equation}
    \mathbb{E}_f[Q_f(z)]=\frac{\mathbb E[N_z]}{2^{2n}}=\frac{2^{2n-m}}{2^{2n}}=\frac{1}{2^m}.
\end{equation}
Using the Chernoff bound~\cite{Chernoff_bound}, we can state the following proposition:
\begin{prop}\label{prop:q_f_close_to_uniform} Let $\varepsilon>0$. Then, for a random function  $f:\{0,1\}^n \times \{0,1\}^n \to \{0,1\}^m$, with probability at least $1-\varepsilon$, a fixed $z\in\{0,1\}^m$ satisfies
    \begin{equation}
        q_f(z)\in\left[\frac{1}{2^m}\pm\frac{\sqrt{3\ln{(2/\varepsilon)}}}{2^{n+m/2}}\right].
    \end{equation}
\end{prop}

From Proposition~\ref{prop:q_f_close_to_uniform}, we see that for a random function $f$, if $n$ is large enough (compared to $m$), then with probability $1-\varepsilon$, the deviation from the uniform distribution $\abs{\delta_f(z)}\leq\frac{\sqrt{3\ln{(2/\varepsilon)}}}{2^{n+m/2}}$ is small. We introduce the set of functions 
\begin{equation}
      \mathcal F_{\varepsilon}:=\left\{f:\{0,1\}^{2n}\rightarrow\{0,1\}^m\mid  q_f(z)\in\left[\frac{1}{2^m}\pm\frac{\sqrt{3\ln{(2/\varepsilon)}}}{2^{n+m/2}}\right]\text{ }
   \forall z\in\{0,1\}^m\right\},
\end{equation}
which, intuitively, corresponds to the functions $f:\{0,1\}^{2n}\rightarrow\{0,1\}^m$ that reproduce a distribution over $\{0,1\}^m$ that is not too far from uniform. Notice that, from \cref{prop:q_f_close_to_uniform}, a random ${f:\{0,1\}^n\times\{0,1\}^n\rightarrow\{0,1\}^m}$ will be in $\mathcal{F}_\varepsilon$ with probability $(1-\varepsilon)^{2^m}$, and, using Bernoulli's inequality, this probability is greater than $1-\varepsilon2^m$, which by properly picking $\varepsilon$, this probability can be made large. We now prove an upper bound for $\omega_\psi^*$ for all these functions in $\mathcal{F}_\varepsilon$.

\begin{lemma}\label{lem:bound w_psi for f in F} 
Let $\varepsilon>0$. Then, for every $f\in\mathcal F_{\varepsilon}$  the following bound holds for \QPVBBfparallel: for every quantum state $\ket{\psi}_{VA'B'}$, with arbitrary dimensional registers $A'$ and $B'$,
\begin{equation}
    \omega_\psi^*\leq \left(\lambda_\gamma\right)^m\big(1+\sqrt{3\ln{(2/\varepsilon)}}2^{-n+m/2}\big).
\end{equation}
\end{lemma}

Notice that the above upper bound is exponentially small in $m$ if $n>(\frac{1}{2}-\log\frac{1}{\lambda_\gamma})m$, i.e., achieving better security requires more classical information than quantum information. 

\begin{proof}
We have that for every quantum state $\ket{\psi}_{VA'B'}$,
\begin{equation}
\begin{split}
       \omega_{\psi}^*&=\max_{\{A^{xy}_a\}_a,\{B^{xy}_b\}_b}\frac{1}{2^{2n} }\sum_{x,y,a}\tr{ \left(M^{f(x,y)}_a \otimes \sum_{a':d_H(a,a')\leq \gamma m}A^{xy}_{a'}\otimes B^{xy}_{a'}\right)\ketbra{\psi}{\psi}}\\
       &=\max_{\{A^{xy}_a\}_a,\{B^{xy}_b\}_b}\sum_{z}\frac{q_f(z)}{n_z}\sum_{x,y:f(x,y)=z}\sum_{a}\tr{ \left(M^{f(x,y)}_a \otimes \sum_{a':d_H(a,a')\leq \gamma m}A^{xy}_{a'}\otimes B^{xy}_{a'}\right)\ketbra{\psi}{\psi}}\\
       &=\max_{\{A^{xy}_a\}_a,\{B^{xy}_b\}_b}\sum_{z}\frac{q_f(z)}{n_z}\sum_{a}\tr{ \left(M^{z}_a \otimes \sum_{\substack{x, y : f(x,y) = z \\ a' : d_H(a, a') \leq \gamma m}}A^{xy}_{a'}\otimes B^{xy}_{a'}\right)\ketbra{\psi}{\psi}}.
\end{split}
\end{equation}
Consider the following upper bound
\begin{equation}\label{eq:upper_bound_max_z}
\begin{split}
     &\sum_{x,y:f(x,y)=z}\tr{ \left(M^{z}_a \otimes \sum_{ a' : d_H(a, a') \leq \gamma m}A^{xy}_{a'}\otimes B^{xy}_{a'}\right)\ketbra{\psi}{\psi}}\\&\leq n_z \max_{x,y:f(x,y)=z}\tr{ \left(M^{z}_a \otimes \sum_{ a' : d_H(a, a') \leq \gamma m}A^{xy}_{a'}\otimes B^{xy}_{a'}\right)\ketbra{\psi}{\psi}},
\end{split}
\end{equation}
then, denoting by $A^z_{a'}$ and $B^z_{a'}$ the corresponding $A^{xy}_{a'}$ and $ B^{xy}_{a'}$ (recall that these $x$ and $y$ are such that $f(x,y)=z$) that attain the maximum in the right-hand side of \eqref{eq:upper_bound_max_z}, we have that 
\begin{equation}\label{eq:w_psi<=A^zB^z}
    \begin{split}
        \omega_{\psi}^*&\leq \max_{\{A^z_a\}_a,\{B^z_b\}_b}\sum_zq_f(z)\sum_a\tr{\left(M^z_a\otimes \sum_{ a' : d_H(a, a') \leq \gamma m}A^{z}_{a'}\otimes B^{z}_{a'}\right)\ketbra{\psi}{\psi}}\\&
        \leq \max_{\{A^z_a\}_a,\{B^z_b\}_b}\sum_z \frac{1}{2^m}\sum_a\tr{\left(M^z_a\otimes \sum_{ a' : d_H(a, a') \leq \gamma m}A^{z}_{a'}\otimes B^{z}_{a'}\right)\ketbra{\psi}{\psi}}\\&+\max_{\{A^z_a\}_a,\{B^z_b\}_b}\sum_z\delta_f(z)\sum_a\tr{\left(M^z_a\otimes \sum_{ a' : d_H(a, a') \leq \gamma m}A^{z}_{a'}\otimes B^{z}_{a'}\right)\ketbra{\psi}{\psi}}\\&
        \leq (\lambda_\gamma)^m+2^m\left(\max_z{\abs{\delta_f(z)}}\right)\max_{\{A^z_a\}_a,\{B^z_b\}_b}\sum_z\sum_a\tr{\left(M^z_a\otimes \sum_{ a' : d_H(a, a') \leq \gamma m}A^{z}_{a'}\otimes B^{z}_{a'}\right)\ketbra{\psi}{\psi}}.
    \end{split}
\end{equation}
Where we used $q_f(z)=\frac{1}{2^m}+\delta_f(z)$ and \cref{coro:w_psi<=...}. 

Since $f\in\mathcal F_\varepsilon$, we have that $\max_z{\abs{\delta_f(z)}}\leq \frac{\sqrt{3\ln{(2/\varepsilon)}}}{2^{n+m/2}}$, and applying again \cref{coro:w_psi<=...}, we have that
\begin{equation}
    \omega_\psi\leq (2^{h(\gamma)}\lambda)^m+(2^{h(\gamma)}\lambda)^m\sqrt{3\ln{(2/\varepsilon)}}2^{-n+m/2}.
\end{equation}
\end{proof}
Now, consider the following subset of $\mathcal{F}_\varepsilon$:
\begin{equation}
      \mathcal F_{\varepsilon}^*:=\left\{f:\{0,1\}^{2n}\rightarrow\{0,1\}^m\mid  q_f(z)\in\left[\frac{1}{2^m}\pm\frac{\sqrt{3\ln{(2/\varepsilon)}}}{2^{n+m/2}}\right]\text{ }
   \forall z\in\{0,1\}^m\text{ with }\sqrt{3\ln{(2/\varepsilon)}}2^{-n+m/2}<2^{-2}\right\}.
\end{equation}
notice that $\sqrt{3\ln{(2/\varepsilon)}}2^{-n+m/2}<2^{-2}$ is not overly restrictive and can be easily achieved by picking $n$ larger than $m/2$. Next, we will show that for any function $f\in\mathcal F^*_\varepsilon$, if a quantum state $\ket{\psi}_{VA'B'}$ is `good' to attack a given basis $z\in\{0,1\}^m$ ---meaning that the probability to successfully attack $z$ is above the bound in \cref{lem:bound w_psi for f in F}, see \cref{def:good_state_for_z}--- then, this state can only be good for a small fraction of all the possible $z$'s. Then, similarly as argued in \cite{bluhm2022single} for \QPVBBf, we will use this to show that the attackers are restricted and, in some sense, they have to decide a small set of possible $z$'s to attack in step 2. of the general attack (before they communicate and learn $z$). 

\begin{definition}\label{def:good_state_for_z}
   Let $\varepsilon,\Delta>0$. We say that a state $\ket{\psi}_{VA'B'}$ is $\Delta-$\emph{good to attack} $z\in\{0,1\}^m$ for \QPVBBfparallel~if there exists \text{POVMs} $\{A^z_a\}_a$ and $\{B^z_b\}_{b}$ acting on $A'$ and $B'$, respectively, such that the probability that the verifiers \emph{accept} on input $z$ (the left hand side of the following inequality) is such that
    \begin{equation}
       \sum_a\tr{ \left(M^{z}_a \otimes \sum_{ a' : d_H(a, a') \leq \gamma m}A^{z}_{a'}\otimes B^{z}_{a'}\right)\ketbra{\psi}{\psi}}\geq (\lambda_\gamma+\Delta)^m\left(1+3\sqrt{3\ln{(2/\varepsilon)}}2^{-n+m/2}\right).
    \end{equation}
\end{definition}

We will see that we will have freedom to choose $\Delta>0$. For now, we only require that $\Delta$ is such that $\lambda_\gamma+\Delta<1$ to ensure that the bound in \cref{def:good_state_for_z} is nontrivial.

\begin{lemma} \label{lem:number_z} Let $\varepsilon,\Delta>0$. Then, for every $f\in\mathcal F^*_\varepsilon$, any quantum state $\ket{\psi}_{VA'B'}$ can be $\Delta-$good for \QPVBBfparallel~on at most a fraction of all the possible $z\in\{0,1\}^m$ given by 
    \begin{equation}
        \left(\frac{\lambda_\gamma}{\lambda_\gamma+\Delta}\right)^m.
    \end{equation}
\end{lemma}

\begin{proof}
    Let $I_\psi=\{z\in\{0,1\}^m\mid \ket{\psi}_{VA'B'} \text{ is }\Delta- \text{good to attack }z\}$. We want to upper bound the size of $I_\psi$. By \cref{lem:bound w_psi for f in F}, see \eqref{eq:w_psi<=A^zB^z},
    \begin{equation}
        \begin{split}
    \left(\lambda_\gamma\right)^m&\left(1+\sqrt{3\ln{(2/\varepsilon)}}2^{-n+m/2}\right)\\&\geq \omega^*_\psi= \max_{\{A^z_a\}_a,\{B^z_b\}_b}\sum_{z\in\{0,1\}^m}q_f(z)\sum_a\tr{\left(M^z_a\otimes \sum_{ a' : d_H(a, a') \leq \gamma m}A^{z}_{a'}\otimes B^{z}_{a'}\right)\ketbra{\psi}{\psi}_{VA'B'}}\\&
    \geq \max_{\{A^z_a\}_a,\{B^z_b\}_b}\sum_{z\in I_\psi}q_f(z)\sum_a\tr{\left(M^z_a\otimes \sum_{ a' : d_H(a, a') \leq \gamma m}A^{z}_{a'}\otimes B^{z}_{a'}\right)\ketbra{\psi}{\psi}_{VA'B'}}\\&
    \geq (\lambda_\gamma+\Delta)^m\left(1+3\sqrt{3\ln{(2/\varepsilon)}}2^{-n+m/2}\right)\sum_{z\in I_\psi}(\frac{1}{2^m}-\sqrt{3\ln{(2/\varepsilon)}}{2^{-n-m/2}}),
        \end{split}
    \end{equation}
    where in the second inequality we just summed over a smaller set of non-negative elements, and the third inequality comes from the hypothesis the $\ket{\psi}_{VAB}$ is $\Delta-$good for $z$ for all $z\in I_{\psi}$. Then, since the element in the summand do not depend on $z$, we have that 
    \begin{equation}
        \abs{I_\psi}\leq \frac{\left(\lambda_\gamma\right)^m\left(1+\sqrt{3\ln{(2/\varepsilon)}}2^{-n+m/2}\right)}{\left(\lambda_\gamma+\Delta\right)^m\frac{1}{2^m}\left(1-\sqrt{3\ln{(2/\varepsilon)}}{2^{-n+m/2}}\right)\left(1+3\sqrt{3\ln{(2/\varepsilon)}}2^{-n+m/2}\right)}\leq\left(\frac{\lambda_\gamma}{\lambda_\gamma+\Delta}\right)^m2^m,
    \end{equation}
    where, since $f\in\mathcal F^*_{\varepsilon}$, then $\sqrt{3\ln{(2/\varepsilon)}}2^{-n+m/2}<2^{-2}$, we used that $\frac{1+x}{1-x}\leq 1+3x$ for $0\leq x<2^{-2}$. 
\end{proof}

In what follows, we will show that, with exponentially high probability, a uniformly drawn function $f\in\mathcal F^*_\varepsilon$ will be such that any $q$-qubit strategy $S=\{\ket\varphi,U^x,V^y,\{A^{xy}_a\}_a,\{B^{xy}_b\}_b\}_{x,y}$ with $q$ linear in $n$ will have exponentially small soundness. The high level idea consists of providing a classical description (up to a certain precision) of $\ket{\varphi}$, $U^x$ and $V^y$, i.e. a classical description of the actions in step 2. of the general attack (before they communicate). We will show that a classical description is `almost as good as' $S$, and we will use this to show that the description allows to recover a set of $z$'s of size at most $2^{(1-\log\frac{\lambda_\gamma+\Delta}{\lambda_\gamma})m}$ for which $f(x,y)$ belongs to. This essentially consists on a (set-valued) compression of $f$, where we relax the condition for the attackers to have a good attack by instead of having to learn the exact value $z=f(x,y)$ they learn set of $z$ containing $f(x,y)$, see \cref{def:rounding}. Then, similarly as in \cite{bluhm2022single}, by using a counting argument with $\delta$-nets, we will see that if $S$ has at least a certain soundness (which is still exponentially small) and $q$ is not large enough (larger than $n$), then, the number of possible compressions will be exponentially smaller than the number of functions $f\in\mathcal{F}^*_\varepsilon$, and therefore attackers, with high probability, will not be able to break the protocol. 

\newpage
\begin{definition}\label{def:rounding}
   Let $\omega_0\in(0,1]$, $\Delta>0$,  $s=1-\log\frac{\lambda_\gamma+\Delta}{\lambda_\gamma}$ and $k_1,k_2,k_3\in\mathbb N$.  A function 
    \begin{equation}
        g:\{0,1\}^{k_1}\times \{0,1\}^{k_2}\times \{0,1\}^{k_3}\rightarrow 
        \mathcal P_{\leq s}(\{0,1\}^m)
    \end{equation}
    is a $(\omega_0,q)$-set-valued classical rounding for \QPVBBfparallel~ of sizes $k_1,k_2,k_3$ if for all functions $f\in\mathcal F^*_\varepsilon$, all $\ell\in\{1,\ldots,2^{2n}\},$ for all $(\omega_0,q,\ell)-$strategies for \QPVBBfparallel, there exist functions ${f_A:\{0,1\}^n\rightarrow\{0,1\}^{k_1}}$, $f_B:\{0,1\}^n\rightarrow\{0,1\}^{k_2}$ and $\lambda\in\{0,1\}^{k_3}$ such that, on at least $\ell$ pairs $(x,y)$,
    \begin{equation}
        f(x,y)\in g(f_A(x),f_B(y),\lambda).
    \end{equation}
\end{definition}

Next, we will construct a set-valued classical rounding using a discretization of a strategy $S$. To this end, we define an approximation of $S$ ---will show that can be constructed with a classical description (discretization) of $S$, see proof of \cref{lem:existence_rounding}---, and we use the following lemmas to prove that an approximation preserves the probabilities induced by $S$ up to a small constant, see \cref{lem:delta_approx is good}. 

\begin{definition}\label{def:approximation} 
Let $\delta\in(0,1)$. A $\delta-$\emph{approximation} of a strategy $S=\{\ket\varphi,U^x,V^y,\{A^{xy}_a\}_a,\{B^{xy}_b\}_b\}_{x,y}$ is the tuple $S_\delta=\{\ket{\varphi_\delta},U^x_\delta,V_\delta^y,\{A^{xy}_a\}_a,\{B^{xy}_b\}_b\}_{x,y}$, where $\ket{\varphi_\delta}$, $U^x_\delta$ and $V_\delta^y$ are such that, for every $x,y\in\{0,1\}^n$,
\begin{equation}
    \norm{\ket{\varphi}-\ket{\varphi_\delta}}_2\leq \delta, \text{ }\norm{U^x-U^x_\delta}_\infty\leq \delta,  \text{ and } \norm{V^y-V_\delta^y}_\infty\leq \delta.  
\end{equation}
\end{definition}
We will use the notation $\ket{\psi^\delta_{xy}}:=U^x_\delta\otimes V^y_\delta\ket{\varphi_\delta}$. 
\begin{lemma}\label{lem:tr<=distance}(Proposition 3.5 in the ArXiv version (v2) of \cite{Escol_Farr_s_2023}). Let $\rho$ and $\sigma$ be two quantum states (density matrices) of the same arbitrary dimension. Then, for every projector $\Pi$, 
\begin{equation}
    \abs{\tr{(\rho-\sigma)\Pi}}\leq \frac{1}{2}\norm{\rho-\sigma}_1\norm{\Pi}_\infty. 
\end{equation}
\end{lemma}

\begin{lemma} \label{lem:P(.)<||.||_2}(Lemma 3.10 in \cite{bluhm2022single}). Let $\ket x$, $\ket y$ be two  unit complex-vectors of the same dimension. Then, 
\begin{equation}
    \mathcal P(\ket x, \ket y)\leq \norm{\ket x-\ket y}_2.
\end{equation}
\end{lemma}
Since the purified distance is an upper bound of the trace distance, we have, as an immediate consequence:
\begin{corollary}\label{coro:1/2||*||_1<=||*||_2}Let $\ket x$, $\ket y$ be two unit complex-vectors of the same dimension. Then, 
    \begin{equation}
         \frac{1}{2}\norm{\ket x-\ket y}_1\leq \norm{\ket x-\ket y}_2.
    \end{equation}
\end{corollary}

\begin{lemma}\label{lem:delta_approx is good}
Let $S=\{\ket\varphi,U^x,V^y,\{A^{xy}_a\}_a,\{B^{xy}_b\}_b\}_{x,y}$ be a $q-$qubit strategy for \QPVBBfparallel.
Then, every $\delta-$\emph{approximation} of $S$, fulfills the following inequality for all $(x,y)$:
\begin{equation}
    \tr{\Pi^{xy}_{AB}\ketbra{\psi^\delta_{xy}}{\psi^\delta_{xy}}}\geq \tr{\Pi^{xy}_{AB}\ketbra{\psi_{xy}}{\psi_{xy}}}-7\delta.
\end{equation}
\end{lemma}

\cref{lem:delta_approx is good} essentially tells us that a $\delta$-approximation of a strategy $S$ does not change much the probabilities induced by $S$, and, therefore,  it captures the essence of it, providing probabilities that are `almost as good as' the original ones. As an immediate consequence we have that for every $\delta$-approximation $S_\delta$ of $S$, 
\begin{equation}\label{eq:w_Sdelta}
\omega_{S_\delta}\geq\omega_S-7\delta.
\end{equation}

\begin{proof} Let $S_\delta$ be a $\delta-$approximation of a $q$-qubit strategy $S=\{\ket\varphi,U^x,V^y,\{A^{xy}_a\}_a,\{B^{xy}_b\}_b\}_{x,y}$. Recall that $\ket{\psi^\delta_{xy}}=U_\delta^x\otimes V_\delta^y\ket{\varphi_\delta}$, for all $x,y\in\{0,1\}^n$. Then, similarly as shown in \cite{bluhm2022single}, 
\begin{equation}\label{eq:distance_psi_psixy}
    \begin{split}
        \frac{1}{2}&\norm{\ketbra{\psi_{xy}}{\psi_{xy}}-\ketbra{\psi^\delta_{xy}}{\psi^\delta_{xy}}}_1\leq  \norm{\ket{\psi_{xy}}-\ket{\psi^\delta_{xy}}}_2= \norm{U^x\otimes V^y\ket{\varphi}-U_\delta^x\otimes V_\delta^y\ket{\varphi_\delta}}_2\\&= \norm{(U^x-U_\delta^x+U_\delta^x)\otimes (V^y-V_\delta^y+V_\delta^y)\ket{\varphi}-U_\delta^x\otimes V_\delta^y\ket{\varphi_\delta}}_2\\&
        \leq 3\delta+3\delta^2+\delta^3\leq 7\delta,
    \end{split}
\end{equation}
where in the first inequality we used \cref{coro:1/2||*||_1<=||*||_2}, in the second inequality we used that\\ ${\norm{X\otimes Y\ket\xi}_2\leq \norm{X}_\infty\norm{Y}_{\infty}\norm{\ket\xi}_2}$, by hypothesis, $\norm{\ket{\varphi}-\ket{\varphi_\delta}}_2\leq \delta, \text{ }\norm{U^x-U^x_\delta}_\infty\leq \delta,$ and \\${\norm{V^y-V_\delta^y}_\infty\leq \delta}$, and in the last inequality we used that $\delta^2,\delta^3\leq \delta$ for $\delta\in(0,1)$. 
Then, sing Lemma~\ref{lem:tr<=distance}
\begin{equation}\label{eq:tr_bounded}
\begin{split}
     &\tr{\Pi^{xy}\left(\ketbra{\psi_{xy}}{\psi_{xy}}-\ketbra{\psi^\delta_{xy}}{\psi^\delta_{xy}}\right)}\leq \abs{\tr{\Pi^{xy}\left(\ketbra{\psi_{xy}}{\psi_{xy}}-\ketbra{\psi^\delta_{xy}}{\psi^\delta_{xy}}\right)}}\\&\leq \frac{1}{2}\norm{\ketbra{\psi_{xy}}{\psi_{xy}}-\ketbra{\psi^\delta_{xy}}{\psi^\delta_{xy}}}_1\norm{\Pi^{xy}}\leq \frac{1}{2}\norm{\ketbra{\psi_{xy}}{\psi_{xy}}-\ketbra{\psi^\delta_{xy}}{\psi^\delta_{xy}}}_1,
\end{split}
\end{equation}
where in the last inequality we used that $A^{xy}_a\preceq\mathbb I_{A'}$ and $\sum_bB^{xy}_b  \preceq \mathbb I_{B'}$, and then we have that \begin{equation}
    \Pi^{xy}_{AB}\preceq \sum_a M^{f(x,y)}_a\otimes \mathbb I_{A'}\otimes\mathbb{I}_{B'}=\mathbb I_{VAB},
\end{equation} 
and thus, $\norm{\Pi^{xy}_{AB}}\leq \norm{\mathbb I_{VA'B'}}=1$. 
Combining \eqref{eq:distance_psi_psixy} and \eqref{eq:tr_bounded}, we have that 
\begin{equation}
    \tr{\Pi^{xy}_{AB}\ketbra{\psi^\delta_{xy}}{\psi^\delta_{xy}}}\geq  \tr{\Pi^{xy}_{AB}\ketbra{\psi_{xy}}{\psi_{xy}}}-7\delta.
\end{equation}
\end{proof}

Now, we have seen that a $\delta$-approximation of a strategy $S$ captures its essence, and we will use it together with $\delta$-nets, to construct a set-valued classical rounding. In order to do so, we will make use of the following lemma.

\begin{lemma}\label{lem:size_net}(Corollary 4.2.13 in \cite{Vershynin_2018}) Let $N\in\mathbb N$ and $\delta>0$. Then, there exists a $\delta$-net, with the Euclidean distance, of the unit sphere in $\mathbb R^N$ with cardinality at most $\left(\frac{3}{\delta}\right)^N$.
\end{lemma}

\begin{lemma}\label{lem:existence_rounding} Let $\varepsilon,\Delta>0$, and $\omega_0\geq (\lambda_\gamma+\Delta)^{m}(1+3\sqrt{3\ln(2/\varepsilon)}2^{-n+m/2})+7\cdot3\Delta^m$. Then, there exists an $(\omega_0,q)$-set-valued classical rounding for \QPVBBfparallel~ of sizes  
\begin{equation}
k_1,k_2\leq\log_2\left(\frac{1}{\Delta}\right)m2^{2q+1}, \text{ and } k_3\leq\log_2\left(\frac{1}{\Delta}\right)m2^{2q+m+1}.
\end{equation}
\end{lemma}

\begin{proof} 
Similarly to Section 4.5.4 in \cite{NielsenChuang}, notice that any state $\ket{\varphi}$ of $2q+m$ qubits can be decomposed as $\ket\varphi=\sum_{j=0}^{2^{2q+m}-1}\varphi_{j}\ket j$ with $\varphi_j\in\mathbb C$ for all $j\in[2^{2q+m}]$ and $1=\sum_{j}\abs{\varphi_{j}}^2=\sum_{j}\text{Re}(\varphi_{j})^2+\text{Im}(\varphi_j)^2$. The latter corresponds to the condition for a point to be on the unit sphere in $\mathbb R^{2\cdot2^{2q+m}}$, i.e. the unit $(2^{2q+m+1}-1)$-sphere and therefore the set of states can be seen as a unit sphere. Similarly, the set of unitary matrices of dimension $d$ can be seen as the unit $(2d^2-1)$-sphere, since for every $U\in\mathcal{U}(d)$, $U U^{\dagger}=\mathbb I_d$, this will correspond to the unitaries that Alice and Bob apply in the step 2. of the general attack.

Let $\delta=3\Delta^m$ and consider a $3\Delta^m$-net $\mathcal N_S$ in Euclidean norm of the $(2^{2q+m+1}-1)$-sphere, which, as argued above, corresponds to the set of quantum states of $2q+m$ qubits, i.e. the set of possible states $\ket\varphi_{VAB}$ that attackers will start in step 2. of the general attack. Moreover, consider $3\Delta^m$-nets $\mathcal N_A$ in and $\mathcal N_B$ in the Schatten $\infty$-norm of the $({2d^2-1})$-sphere, where $d=2^q$, which, also as argued above, correspond to the set of unitary operators that Alice and Bob apply in step 2. of the general attack, respectively. Pick the these $\Delta^m$-nets such that their cardinalities are at most $(3/\Delta^m)^{2^{2q+m+1}}$, $(3/\Delta^m)^{2d^2}$ and, $(3/\Delta^m)^{2d^2}$, respectively, which exist due to \cref{lem:size_net}.  

We now construct a an $(\omega_0,q)$-set-valued classical rounding, whose sizes, as argued above, are  of size at most $k_1=k_2=\log_2\left(\frac{1}{\Delta}\right)m2^{2q+1}$, $k_3=\log_2\left(\frac{1}{\Delta}\right)m2^{2q+m+1}$. Let $S=\{\ket\varphi,U^x,V^y,\{A^{xy}_a\}_a,\{B^{xy}_b\}_b\}_{x,y}$ be an $(\omega_0,q,\ell)-$ strategy for \QPVBBfparallel, we define
\begin{itemize}
    \item $\lambda$ as the element in $\mathcal N_S$ that is closest to $\ket\varphi$ in Euclidean norm, and denote by $\ket{\varphi_\delta}$ the state described by  $\lambda$,
    \item $f_A(x)$ as the element in $\mathcal N_A$ that is closest to $U^x$ in operator norm, and enote by $U^x_\delta$ the unitary described by  $f_A(x)$,
    \item $f_B(y)$ as the element in $\mathcal N_B$ that is closest to $V^y$ in operator norm, and denote by $V^y_\delta$ the unitary described by  $f_A(y)$.
\end{itemize}
If the closest element is not unique, make an arbitrary choice. Let $\ket{\psi^\delta_{xy}}=U^x_\delta\otimes V^y_{\delta}\ket{\varphi_\delta}$. By contruction, 
\begin{equation}
    \norm{\ket{\varphi}-\ket{\varphi_\delta}}_2\leq \delta, \text{ }\norm{U^x-U^x_\delta}_\infty\leq \delta,  \text{ and } \norm{V^y-V_\delta^y}_\infty\leq \delta, 
\end{equation}
and therefore,  $S_\delta=\{\ket{\varphi_\delta},U^x_\delta,V^y_\delta,\{A^{xy}_a\}_a,\{B^{xy}_b\}_b\}_{x,y}$ is a $\delta$-approximation of $S$. 
Now, define 
\begin{equation}\label{eq;def_g}
    g(f_A(x),f_B(y),\lambda):=\{z\mid \exists \{A^z_a\}_a,\{B^z_b\}_b \text{ with } \sum_a\tr{M^z_a\otimes \sum_{ a' : d_H(a, a') \leq \gamma m}A^{z}_{a'}\otimes B^{z}_{a'}\ketbra{\psi_{xy}^\delta}{\psi_{xy}^\delta}}\geq \omega_0-7\delta \}.
\end{equation}
Since, by hypothesis,  $\sqrt{3\ln{(2/\varepsilon)}}2^{-n+m/2}<\frac{1}{4}$,  and $f\in\mathcal F_\varepsilon^*$, by \cref{lem:number_z}, the right-hand side of \eqref{eq;def_g} has cardinality at most $2^{(1-\log(\frac{\lambda_\gamma+\Delta}{\lambda_\gamma}))m}$. 
    We want to show that $g$ is a $(\omega_0,q)$-set-valued classical rounding (the sizes $k_1,k_2$, and $k_3$ are already bounded). Consider a $(\omega_0,q,\ell)-$strategy, then, there exists a set $\mathcal B\subseteq\{0,1\}^{2n}$ with $\abs{\mathcal{B}}\geq \ell$ such that for all $(x,y)\in\mathcal B$, 
    \begin{equation}
        \sum_a\tr{M^{f(x,y)}_a\otimes \sum_{ a' : d_H(a, a') \leq \gamma m}A^{xy}_{a'}\otimes B^{xy}_{a'}\ketbra{\psi_{xy}}{\psi_{xy}}}\geq \omega_0.
    \end{equation}
Then, since $S_\delta$ is a $\delta$-approximation of $S$, by \cref{lem:delta_approx is good}, we have that, for all $(x,y)\in\mathcal B$
 \begin{equation}
        \sum_a\tr{M^{f(x,y)}_a\otimes \sum_{ a' : d_H(a, a') \leq \gamma m}A^{xy}_{a'}\otimes B^{xy}_{a'}\ketbra{\psi_{xy}^\delta}{\psi_{xy}^\delta}}\geq \omega_0-7\delta, 
    \end{equation}
since $\abs{\mathcal B}\geq\ell$, we have that 
\begin{equation}
    f(x,y)\in g(f_A(x),f_B(y),\lambda),
\end{equation}
    on at least $\ell$ pairs $(x,y)$. 
\end{proof}

A $(\omega_0,q)$-set-valued classical rounding $g$ for \QPVBBfparallel, defined in \cref{def:rounding}, `covers'\\ $(\omega_0,q,\ell)-$strategies for \QPVBBfparallel. Next, we will show that there exists $g$ such that if one has $f(x,y)\in g(f_A(x),f_B(y),\lambda)$ on a fraction $\beta$ of all the possible inputs $(x,y)$, i.e. those pairs for which attackers prepared a `good' attack (success probability of at least $\omega_0$), then the number of qubits $q$ that Alice and Bob pre-share grows with both $\beta$ and $n$. This means that the more pairs $(x,y)$ the attackers have to prepare a `good' attack, the more qubits they need to pre-share. In particular, in the following lemma, we show that $q$ grows logarithmically in $\beta$ and linearly in $n$.

\begin{lemma}\label{lem:q_bounded_in_rounding}
Let $\varepsilon>0$, $\beta\in(0,1]$, and $\omega_0\geq(\lambda_\gamma+\Delta)^{m}(1+3\sqrt{3\ln(2/\varepsilon)}2^{-n+m/2})+7\cdot3\Delta^m$. Fix an $(\omega_0,q)$-set-valued classical rounding $g$ for \QPVBBfparallel~of sizes   $k_1,k_2\leq\log_2\left(\frac{1}{\Delta}\right)m2^{2q+1}$, $k_3\leq\log_2\left(\frac{1}{\Delta}\right)m2^{2q+m+1}$. Let  $f\in \mathcal F^*_\varepsilon$ be such that for any $f_A,f_B$ and $\lambda$ as defined in \cref{def:rounding}, $f(x,y)\in g(f_A(x),f_B(y),\lambda)$ holds on more than $\beta\cdot 2^{2n}$ pairs $(x,y)$, then  with probability at least $1-2^{-m2^{n-\log(1/\beta)}}$, $f$ is such that 
\begin{equation}\label{eq:2^q<beta...}
    \log\left(\frac{1}{\Delta}\right)2^{2q+2}(1+2^{-n+m-1})\geq \beta \log\left(\frac{\lambda_\gamma+\Delta}{\lambda_\gamma}\right)2^n+\frac{1}{m}2^{-n+m}\log(1-\varepsilon).
    \end{equation}
\end{lemma}

\begin{proof}
By \cref{lem:existence_rounding}, there exists an $(\omega_0,q)$-set-valued classical rounding $g$ of sizes\\ $k_1,k_2\leq\log_2\left(\frac{1}{\Delta}\right)m2^{2q+1},\log_2\left(\frac{1}{\Delta}\right)m2^{2q+1}$, $k_3\leq\log_2\left(\frac{1}{\Delta}\right)m2^{2q+m+1}$. The number of possible functions $g(f_A(x),f_B(y),\lambda)$ that Alice and Bob can implement depends on the number of choices of ${f_A:\{0,1\}^n\rightarrow\{0,1\}^{k_1}}$, $f_B:\{0,1\}^n\rightarrow\{0,1\}^{k_2}$ and $\lambda\in\{0,1\}^{k_3}$. In total, there are $(2^{k_1})^{2^n}\cdot(2^{k_2})^{2^n}\cdot(2^{k_3})$ such functions. By hypothesis, $f(x,y)\in g(f_A(x),f_B(y),\lambda)$ on at least $\beta\cdot 2^{2n}$ pairs $(x,y)$, denote by $\mathcal{B}$ the set of these $(x,y)$, and, recalling that the cardinality of the set $g(f_A(x),f_B(y),\lambda)$ is at most $2^{(1-\log\frac{\lambda_\gamma+\Delta}{\lambda_\gamma})m}$, we have that, given $g$, the total number of ways to assign outputs for these pairs is $(2^{(1-\log\frac{\lambda_\gamma+\Delta}{\lambda_\gamma})m})^{\beta 2^{2n}}$. For the remaining $(1-\beta)\cdot 2^{2n}$ pairs of $(x,y)$, no compression is applied (i.e., we do not have the guarantee $f(x,y)\in g(f_A(x),f_B(y),\lambda)$). In these cases, we have that $f(x,y)\in\{0,1\}^m$, for which we have $(2^m)^{(1-\beta)2^{2n}}$ possible ways to assign values. 

On the other hand, we have seen that the cardinality of $\mathcal{F}^*$ is $(1-\varepsilon)^{2^m}2^{m2^{2n}}$.
Then, we have that, using that $f\in\mathcal F^*_\varepsilon$ is drawn uniformly at random,
\begin{equation}
\begin{split}
 &\text{Pr}\{f\in\mathcal F^*_\varepsilon: \exists f_A,f_B,\lambda \text{ s.t. } f(x,y)\in g(f_A(x),f_B(y),\lambda)\text{ } \forall (x,y)\in \mathcal B\}\\
 &=\frac{\abs{f\in\mathcal F^*_\varepsilon:\exists f_A,f_B,\lambda \text{ s.t. } f(x,y)\in g(f_A(x),f_B(y),\lambda)\text{ } \forall (x,y)\in \mathcal B}}{\abs{\mathcal{F}^*_\varepsilon}}\\
 &\leq \frac{\left(2^{\log_2\left(\frac{1}{\Delta}\right)m2^{2q+1}}\right)^{2^n}\cdot\left(2^{\log_2\left(\frac{1}{\Delta}\right)m2^{2q+1}}\right)^{2^n}\cdot\left(2^{\log_2\left(\frac{1}{\Delta}\right)m2^{2q+m+1}}\right)\cdot(2^{(1-\log(\frac{\lambda_\gamma+\Delta}{\lambda_\gamma}))m})^{\beta 2^{2n}}\cdot(2^m)^{(1-\beta)2^{2n}}}{(1-\varepsilon)^{2^m}2^{m2^{2m}}}\\
 &=2^{\log(\frac{1}{\Delta})m2^{2q+n+2}(1+2^{-n+m-1})+2^m\log(\frac{1}{1-\varepsilon})-\beta\log(\frac{\lambda_\gamma+\Delta}{\lambda_\gamma})m2^{2n}}.
\end{split}
\end{equation}
Notice that the above quantity will be decreasing in $m$ and $n$ if the `dominating' term is the negative one, i.e. 
\begin{equation}\label{eq:condition 1 with high prob}
    \log(\frac{1}{\Delta})m2^{2q+n+2}(1+2^{-n+m-1})+2^m\log(\frac{1}{1-\varepsilon})<\beta\log(\frac{\lambda_\gamma+\Delta}{\lambda_\gamma})m2^{2n},
    \end{equation}
which is the converse of condition \eqref{eq:2^q<beta...}. In particular, we have that if \eqref{eq:condition 1 with high prob} holds,
\begin{equation}
    2^{\log(\frac{1}{\Delta})m2^{2q+n+2}(1+2^{-n+m-1})+2^m\log(\frac{1}{1-\varepsilon})-\beta\log(\frac{\lambda_\gamma+\Delta}{\lambda_\gamma})m2^{2n}}<2^{-m2^{n-\log(\frac{1}{\beta})}}.
\end{equation}

\end{proof}

From \eqref{eq:2^q<beta...}, we see that, picking $n>m$ and small $\varepsilon$, the terms $2^{-n+m-1}$ and $\frac{1}{m}2^{-n+m}\log(1-\varepsilon)$ become negligible,  and in order to have a `good' attack (i.e. $f(x,y)\in g(f_A(x),f_B(y),\lambda)$) for at least $\beta\cdot 2^{2n}$ possible $(x,y)$'s,  the inequality \eqref{eq:2^q<beta...} becomes: "$2^{2q}$ is approximately greater or equal to $\beta2^n$",  which implies that 
\begin{equation}
    2q	\succsim n-\log(1/\beta) \hspace{5mm}\text{ (up to constant factors)}.
\end{equation}
However, we do not have control over the number of pairs that attackers have prepared a good attack for. The following lemma states that if attackers have prepared an strategy that has at least a certain soundness, then, there must be a number of pairs $(x,y)$ for which they prepared a good attack. 

\begin{lemma}\label{lem:at leas b pairs are good}
    Let $\omega_1\in(0,1]$, and $S=\{\ket\varphi,U^x,V^y,\{A^{xy}_a\}_a,\{B^{xy}_b\}_b\}_{x,y}$ be a $q$-qubit strategy for \QPVBBfparallel~ such that $\omega_S\geq \omega_1$. Then, for $\omega_0<\omega_1$, there exist at least $\frac{\omega_1-\omega_0}{1-\omega_0}2^{2n}$ of pairs $(x,y)$ such that
    \begin{equation}
    \tr{\Pi^{xy}_{AB}\ketbra{\psi_{xy}}{\psi_{xy}}}\geq \omega_0,
    \end{equation}
this is, $S$ is an $(\omega_0,q,\frac{\omega_1-\omega_0}{1-\omega_0}2^{2n})$-strategy. 
\end{lemma}
\begin{proof}
    Let $J:=\{(x,y)\mid \tr{\Pi^{xy}_{AB}\ketbra{\psi_{xy}}{\psi_{xy}}}\geq \omega_0 \}$, we want to find a lower bound on the cardinality of $J$, and denote by $J^c$ its complementary set. We have that 
    \begin{equation}
    \begin{split}
        \omega_1&\leq \omega_S=\frac{1}{2^{2n}}\sum_{x,y} \tr{\Pi^{xy}_{AB}\ketbra{\psi_{xy}}{\psi_{xy}}}\\&=\frac{1}{2^{2n}}\sum_{(x,y)\in J} \tr{\Pi^{xy}_{AB}\ketbra{\psi_{xy}}{\psi_{xy}}}+\frac{1}{2^{2n}}\sum_{(x,y)\in J^c} \tr{\Pi^{xy}_{AB}\ketbra{\psi_{xy}}{\psi_{xy}}}
        \leq \frac{1}{2^{2n}}\abs{J}+\frac{1}{2^{2n}}\omega_0\abs{J^c},
    \end{split}
    \end{equation}
    where the first inequality holds by hypothesis, and in the last inequality we used the trivial bound $\tr{\Pi^{xy}_{AB}\ketbra{\psi_{xy}}{\psi_{xy}}}\leq 1$ for $(x,y)\in J$ and we used the bound $\tr{\Pi^{xy}_{AB}\ketbra{\psi_{xy}}{\psi_{xy}}}\leq \omega_0$ for ${(x,y)\in J^c}$. Then, using that $\abs{J^c}=2^{2n}-\abs{J}$, we have that $\abs{J}\leq \frac{\omega_1-\omega_0}{1-\omega_0}2^{2n}.$
\end{proof}

\begin{theorem}\label{thm:main_bb84_parallel_soundess} 
Let $n>m$, $\varepsilon\leq 2^{-m-1}$ and $\Delta>0$. For every $c<1$, with probability at least $1-2^{-m2^{n-c m \log(\frac{1}{\lambda_\gamma+\Delta})}}$, a uniformly random $f\in \mathcal F^*_\varepsilon$ will be such that, if the number of qubits $q$ that the attackers pre-share to attack \QPVBBfparallel~is such that
   \begin{equation}\label{eq:soundess_theorem}
        2q< n-c m\log\left(\frac{1}{\lambda_\gamma+\Delta}\right)+\log\frac{(1-(\lambda_\gamma+\Delta)^{1-c})\log\left(\frac{\lambda_\gamma+\Delta}{\lambda_\gamma}\right)}{8\log(1/\Delta)},
    \end{equation}
     then, the probability that the verifiers \emph{accept} is at most
     \begin{equation}\label{eq:w_theorem}
        \left((\lambda_\gamma+\Delta)^{c}\right)^m(1+3\sqrt{3\ln(2/\varepsilon)}2^{-n+m/2})+7\cdot 3\Delta^m.
    \end{equation}
\end{theorem}

Notice that the bound in \cref{thm:main_bb84_parallel_soundess} exponentially decays in $m$ if $\lambda_\gamma+\Delta<1$. Moreover, since, by hypothesis $\varepsilon\leq 2^{-m-1}$, in particular we have that, under the conditions of \cref{thm:main_bb84_parallel_soundess}, any $q$-qubit strategy $S$ for \QPVBBfparallel~is such that 
\begin{equation}
    \omega_S\leq\left((\lambda_\gamma+\Delta)^{c}\right)^m(1+3\sqrt{3m\ln(2)}2^{-n+m/2})+7\cdot 3\Delta^m.
\end{equation}
\cref{thm:main_bb84_parallel_soundess} leaves freedom to pick the values $\Delta$ and $c$. If one wants a lower upper bound on the soundness, these should be picked small and big, respectively. By picking $\Delta$ small enough, e.g, $\Delta=10^{-5}$, the term $\lambda_\gamma+\Delta$ is strictly smaller than 1 for an error $\gamma$ up to roughly $3.6\%$ and we have that up to that error, the upper bound on the soundness in \cref{thm:main_bb84_parallel_soundess} will decay exponentially. Notice that the asymptotic behavior of the upper bound on the soundness behaves as 
\begin{equation}
    \left((\lambda_\gamma+\Delta)^{c}\right)^m.
\end{equation}

\begin{proof}
    Let $S$ be a $q$-qubit strategy $S$ such that 
    \begin{equation}\label{eq:w_S>= in proof}
        \omega_S\geq  (\lambda_\gamma+\Delta)^{c m}(1+3\sqrt{3\ln(2/\varepsilon)}2^{-n+m/2})+7\cdot 3\Delta^m,
    \end{equation}
    and let $\omega_0=(\lambda_\gamma+\Delta)^{m}(1+3\sqrt{3\ln(2/\varepsilon)}2^{-n+m/2})+7\cdot 3\Delta^m$, then, by \cref{lem:at leas b pairs are good}, 
    $S$ is an $(\omega_0,q,\beta\cdot 2^{2n})$-strategy, with 
    \begin{equation}
       \beta=\frac{(\lambda_\gamma+\Delta)^{cm}(1+3\sqrt{3\ln(2/\varepsilon)}2^{-n+m/2})(1-(\lambda_\gamma+\Delta)^{(1-c)m})}{1-(\lambda_\gamma+\Delta)^{m}(1+3\sqrt{3\ln(2/\varepsilon)}2^{-n+m/2})}.
    \end{equation}
Since $3\sqrt{3\ln(2/\varepsilon)}2^{-n+m/2}\geq0$, we have that 
\begin{equation}
    \beta\geq\frac{(\lambda_\gamma+\Delta)^{cm}(1-(\lambda_\gamma+\Delta)^{(1-c)m})}{1-(\lambda_\gamma+\Delta)^{m}(1+3\sqrt{3\ln(2/\varepsilon)}2^{-n+m/2})}\geq \frac{(\lambda_\gamma+\Delta)^{cm}(1-(\lambda_\gamma+\Delta)^{1-c})}{1-(\lambda_\gamma+\Delta)^{m}(1+3\sqrt{3\ln(2/\varepsilon)}2^{-n+m/2})},
\end{equation}
where we used that $1-(\lambda_\gamma+\Delta)^{cm}\geq1-(\lambda_\gamma+\Delta)^{1-c}$. Then, using the inequality $\frac{1}{1-x}\geq 1$ for $x\in(0,1)$, 
\begin{equation}
    \beta\geq (\lambda_\gamma+\Delta)^{cm}(1-(\lambda_\gamma+\Delta)^{1-c})=:\beta_0,
\end{equation}
then, in particular, $S$ is a $(\omega_0,q,\beta_0\cdot 2^{2n})$-strategy. Then, by \cref{lem:existence_rounding}, there exist an $(\omega_0,q)$-set-valued classical rounding of sizes  $k_1,k_2\leq\log_2\left(\frac{1}{\Delta}\right)m2^{2q+1}$, $k_3\leq\log_2\left(\frac{1}{\Delta}\right)m2^{2q+m+1}$.

Let $f\in\mathcal{F}^*_\varepsilon$ be such that  $f(x,y)\in g(f_A(x),f_B(y),\lambda)$ holds on more than $\beta_0\cdot 2^{2n}$ pairs $(x,y)$ for any $f_A,f_B$ and $\lambda$, by the counterstatement of \cref{lem:q_bounded_in_rounding}, a uniformly random  $f\in\mathcal{F}^*_\varepsilon$, with probability at least $1-2^{-m2^{n-\log(1/\beta)}}$,  will be such that 
\begin{equation}
    \log\left(\frac{1}{\Delta}\right)2^{2q+2}(1+2^{-n+m-1})\geq \beta \log\left(\frac{\lambda_\gamma+\Delta}{\lambda_\gamma}\right)2^n+\frac{1}{m}2^{-n+m}\log(1-\varepsilon).
    \end{equation}
Since $n>m$, we have that $1\geq 2^{-n+m-1}$, and, using that $\varepsilon\leq 2^{-m-1}$, the last summand above is such that 
\begin{equation}\label{eq:term_vanishing}
    \frac{1}{m}2^{-n+m}\log(1-\varepsilon)\geq -\frac{1}{m}2^{-n},
\end{equation}
    where we used that $-\log(1-x)\geq 2x$ for $x\leq \frac{1}{2}$, therefore \eqref{eq:term_vanishing} is exponentially decreasing in $n$ and then, we have
\begin{equation}\label{eq:q> in proof}
    \log\left(\frac{1}{\Delta}\right)m2^{2q+3}\geq m \beta_0 \log\left(\frac{\lambda_\gamma+\Delta}{\lambda_\gamma}\right)2^n,
    \end{equation}
    and therefore, 
    \begin{equation}
        2q+3\geq n-cm\log\left(\frac{1}{\lambda_\gamma+\Delta}\right)+\log\left(1-(\lambda_\gamma+\Delta)^{1-c}\right)+\log\log\left(\frac{\lambda_\gamma+\Delta}{\lambda_\gamma}\right)-\log\log\frac{1}{\Delta}.
    \end{equation}
We have seen that, with probability at least $1-2^{-m2^{n-\log(1/\beta)}}$, a uniformly random $f\in\mathcal{F}^*_\varepsilon$ with \eqref{eq:w_S>= in proof} implies \eqref{eq:q> in proof}. However, by hypothesis, we have strict inequality in the other direction in \eqref{eq:q> in proof}, and therefore, this implies \eqref{eq:soundess_theorem}. 
\end{proof}

\subsection{Improved error-tolerance for \QPVBBf}
In \cite{bluhm2022single}, it was shown that \QPVBBf~ is secure for attackers who pre-share a linear amount (in $n$) of qubits as long as the error remains below $2\%$. Here, by considering the case $m=1$ in \QPVBBfparallel, which corresponds to \QPVBBf, we show that the protocol can tolerate an error almost up to $14,6\%$, presenting an order-of-magnitude improvement in error tolerance.

For the case of $m=1$, the verifiers accept if, in step 4. of the description of \QPVBBf, $a=v$, i.e. if they received the correct outcome. Then, applying \cref{thm:main_bb84_parallel_soundess} for $m=1$ and, recall that since the acceptance criterion is binary, $\lambda_\gamma=\lambda_0=(\frac{1}{2}+\frac{1}{2\sqrt{2}})$, picking  $\Delta=10^{-5}$ and $c=0.999$, then we have the following corollary:

\begin{corollary}
\label{coro:1_round} Let $n,m\in\mathbb N$, with $n>m$ and $n\geq 36$, and $\varepsilon\leq 2^{-m-1}$. Then, with probability at least $1-2^{-2^{n-c  \log(\frac{1}{\lambda_0+\Delta})}}$, a uniformly random $f\in \mathcal F^*_\varepsilon$ will be such that, if 
   \begin{equation}\label{eq:soundess_1_round}
        q< \frac{1}{2}n+\frac{1}{2}\log\frac{(\lambda_0+\Delta)^c(1-(\lambda_0+\Delta)^{1-c})\log\left(\frac{\lambda_0+\Delta}{\lambda_0}\right)}{8\log(1/\Delta)}\simeq\frac{1}{2}n-17.8797,
    \end{equation}
     any $q$-qubit strategy $S$ for \QPVBBf~is such that 
     \begin{equation}\label{eq:w_1_round}
         \omega_S\leq\left(\frac{1}{2}+\frac{1}{2\sqrt{2}}+\Delta\right)^{c}(1+3\sqrt{6\ln(2)}2^{-n})+7\cdot 3\Delta\simeq 0.853699(1+3\sqrt{6\ln(2)}2^{-n})+0.00021.
    \end{equation}
\end{corollary}

Thus, the upper bound in \eqref{eq:w_1_round} converges exponentially in $n$ to 
\begin{equation}
    0.853909\ldots.
\end{equation}
Notably, the attack described in \cref{rmk:Breidbart_attack} achieves a success probability of $\frac{1}{2}+\frac{1}{2\sqrt{2}}=0.85355\ldots$, showing that our bound is essentially tight. This implies that even if Alice and Bob share a linear amount $q=O(n)$ of pre-shared qubits, they cannot outperform an attack that relies on no pre-shared entanglement.\\

\noindent\textbf{Almost tight result for error-free case}\\

We have shown after \cref{thm:main_bb84_parallel_soundess} that the asymptotic behavior of the soundness of \QPVBBfparallel is given by $\left((\lambda_\gamma+\Delta)^{c}\right)^m$. Similarly as above, picking $\Delta=10^{-5}$ and $c=0.999$, the upper bound for the free-error case scales asymptotically as $(0.853699...)^m$, which is almost achieved by the attack described in \cref{rmk:Breidbart_attack} that has winning probability of $(\frac{1}{2}+\frac{1}{2\sqrt 2})^m=(0.85355\ldots)^m$, which recall that uses no-preshared entanglement.\\

\section{Parallel repetition of \routingfparallel}

In this section, we study the $m$-fold parallel repetition of \routingf, which we denote by \routingfparallel. Similarly as in \cref{sec:parallel_f-BB84}, we will describe the protocol and its general attack. Due to the similarities that both \QPVBBfparallel~and \routingfparallel~present, we will use similar techniques as in \cref{sec:parallel_f-BB84}. 

\begin{definition} \label{def qpv routing f} \emph{(\routingfparallel~protocol)}.
Let $n,m\in\mathbb{N}$, $f:\{0,1\}^n \times \{0,1\}^n \to \{0,1\}^m$, and consider an error parameter $\gamma\in[0,\frac{1}{2})$. The \routingfparallel~protocol is described as follows:
\begin{enumerate}
    \item The verifiers $V_0$ and $V_1$ secretly agree on random bits $x,y\in\{0,1\}^n$ and $m$ BB84 states uniformly at random, i.e. ${\ket{\phi_i}\in\{\ket0,\ket1,\ket+,\ket-\}}$ for $i\in[m]$. 
    \item Verifier $V_0$ sends $\otimes_{i=1}^m\ket{\phi_i}$ and $x\in\{0,1\}^n$ to $P$, and $V_1$ sends $y\in\{0,1\}^n$ to $P$ so that all the information arrives at $pos$ simultaneously. The classical information is required to travel at the speed of light, the quantum information can be sent arbitrarily slow.
    \item Immediately, for all $i\in[m]$, $P$ sends the $i$th qubit to the verifier $V_{z_i}$, with $z_i:=f(x,y)_i$. The qubits are required to be sent back to the verifiers at the speed of light.
    \item  Upon receiving the qubits from the prover, verifier $V_0$ ($V_1)$ performs projective measurements onto $\ket{\phi_i}$ for all $i$ such that $f(x,y)_i=0$ ($f(x,y)_i=1$). Let $a\in\{0,1\}^m$ with $a_i=0$ if the projective measurement yields to the correct outcome, and $a_i=1$, otherwise, for all $i\in[m]$. If all the qubits arrive at the time consistent with $pos$, and $w_H(a)\leq\gamma m$  (consistency with the error), the verifiers \emph{accept}. Otherwise, they \emph{reject}.       
\end{enumerate}
\end{definition}

See Fig.~\ref{fig:parallel_BB84} for a schematic representation of the \routingfparallel~protocol. The \routing~ and \routingparallel{m}~(its $m$-fold parallel repetition) protocols are recovered 
if the only classical information that is sent from the verifiers is $y\in\{0,1\}$ and $y\in\{0,1\}^m$, respectively (and $z=y$), and \routingf~is recovered by setting $m=1$.

\begin{figure}[h]
    \centering
    \scalebox{0.9}{
    \begin{tikzpicture}[node distance=3cm, auto]
    \node (A) {$V_0$};
    \node [left=1cm of A] {};
    \node [right=2.5cm of A] (state) {};
    \node [right=of A] (B) {$P$};
    \node [right=of B] (C) {$V_1$};
    \node [right=1cm of C] {};
    \node [below=of A] (D) {};
    \node [below=of B] (E) {};

    \node [above=-0.1cm of A] (N) {};
    \node [right=0.8cm of N] {$\diagdots[120]{2.5em}{0.1em}$};
    
    \node [right=0.8cm of A] (M) {};
    \node [left=1.5cm of C] (M2) {};
    
    \node [below=of C] (F) {};
    \node [below=of D] (G) {};%$V_0$
    \node [below=of E] (H) {};
    \node [below=of F] (I) {};%$V_1$
    \node [left= 6cm of E] (J) {};
    \node [below= 3cm of J] (K) {};
    \node [above= 3cm of J] (L) {};
    \node [right=0.1 of E](P fxy){$f(x,y)=z$};

    \draw [->, transform canvas={xshift=0pt, yshift=0pt}, quantum] (M) -- (E) node[midway] (x) {} ;
    \draw [->] (A) -- (E);
    \draw [->] (C) -- (E);
    \draw [][->,quantum] (E) -- (I) node[midway] (q) {$\otimes_{i:z_i=0}\ket{\phi_i}$}; %line width=0.4mm, dotted
    \draw [][->,quantum] (E) -- (G); %line width=0.4mm, dotted

    \draw [->] (L) -- (K) node[midway] {time};

    \node[below=0.5cm of state] {$\otimes_{i=1}^m\ket{\phi_i}$};
    
    \node[left=1.4cm of x, transform canvas={xshift=+ 2pt, yshift = +2 pt}] {$x\in\{0,1\}^n$};
    \node[right = 2.9cm of x, transform canvas={xshift=+ 2pt, yshift = +2 pt}] {$y \in \{0,1\}^n$};
    \node[left = 3.5cm of q] {$\otimes_{i:z_i=1}\ket{\phi_i}$};

    \node [above=0.5cm of A] (posV00) {};
    \node [left=1cm of posV00] (posV0) {};
    \node [above=0.5cm of C] (posV11) {};
    \node [right=1cm of posV11] (posV1) {};
    \draw [->] (posV0) -- (posV1) node[midway] {position};

    \node [right=0.8cm of C] (VP0) {};
    \node [right=4.55cm of E] (VPP) {};
    \node [right=0.8cm of I] (PV) {};
    \end{tikzpicture}
    }
\caption{Schematic representation of the \routingfparallel~protocol. Undulated lines represent quantum information, whereas straight lines represent classical information. The slowly travelling quantum system $\otimes_{i=1}^m\ket{\phi_i}$ originated from $V_0$ in the past.}
\label{fig:parallel_routing}
\end{figure}

Analogous to the security analysis of \QPVBBfparallel, we will consider the purified version of \routingfparallel, which is equivalent to it. The difference relies on, instead of $V_0$ sending BB84 states, $V_0$ prepares $m$  EPR pairs $\ket{\Phi^+}_{V_0^1Q_1}\otimes\dots\otimes\ket{\Phi^+}_{V_0^mQ_m}$ and sends the registers $Q_1\ldots Q_m$ to the prover. 

Upon receiving back the qubits, verifier $V_0$ performs the measurement $\{N^z_a\}_{a\in\{0,1\}^m}$, where
\begin{equation}
     N^{z}_a:=\bigotimes_{i\in[m]}N^{z_i}_{a_i},
\end{equation}
with
\begin{equation}
    N^{z_i}_{a_i}:=\begin{cases}
        \ketbra{\Phi^+}{\Phi^+}_{V_0^{z_i}P^{z_i}},&\text{ if }a_i=0,\\
        \mathbb I-\ketbra{\Phi^+}{\Phi^+}_{V_0^{z_i}P^{z_i}}, &\text{ if }a_i=1.
    \end{cases}
\end{equation}
%\begin{equation}     N^{z}:=\bigotimes_{i\in[m]}\ketbra{\Phi^+}{\Phi^+}_{V_0^{z_i}P^{z_i}},\end{equation}
where $P^{z_i}$ denotes the register the qubit register of $P$ that she send to the verifier $z_i$. In this way, the verifiers delay the choice of basis in which the $m$ qubits are encoded, which, in contrast to the above prepare-and-measure version, will make any attack independent of the state that was sent. 

The most general attack to \routingfparallel~is described as follows (note that steps 1. and 2. are the same as in \QPVBBfparallel, since in both protocols the verifiers send BB84 states and classical information): 
\begin{enumerate} \item Alice intercepts the $m$ qubit state $Q_1\ldots Q_m$ and applies an arbitrary quantum operation to it and to her a local register that she possess, possibly entangling them. She keeps part of the resulting state, and sends the rest to Bob. Since the qubits $Q_1\ldots Q_m$  can be sent arbitrarily slow by $V_0$ (the verifiers only time the classical information), this happens before Alice and Bob can intercept $x$ and $y$.  

\item Alice intercepts $x$ and Bob intercepts $y$. At this stage, Alice, Bob, and $V_0$ share a quantum state $\ket{\varphi}$, make a partition and let $q$ be the number of qubits that Alice and Bob each hold, recall that $m$ qubits are held by $V_0$ and thus the three parties share a quantum state $\ket{\varphi}$ of $2q+m$ qubits.  Alice and Bob apply a unitary $U_{A_\text{k}A_\text{c}}^{x}$ and $V_{B_\text{k}B_\text{c}}^{y}$ on their local registers $A_\text{k}A_\text{c}=:A$ and $B_\text{k}B_\text{c}=:B$, both of dimension $d=2^{q}$, where k and c denote the registers that will be kept and communicated, respectively. Due to the Stinespring dilation, we consider unitary operations instead of quantum channels. They end up with the quantum state ${\ket{\psi_{xy}}=\mathbb{I}_{V}\otimes U_{A_\text{k}A_\text{c}}^x\otimes V_{B_\text{k}B_\text{c}}^y\ket\varphi}$. Alice sends register $A_\text{c}$ and $x$ to Bob (and keeps register $A_\text{k}$), and Bob sends register $B_\text{c}$ and $y$ to Alice (and keeps register $B_\text{k})$. 
\item Alice and Bob perform unitaries  $K^{xy}$ and $L^{xy}$ on their local registers $A_\text{k}B_\text{c}=:A'$ and $B_\text{k}A_\text{c}=:B'$. The registers $A'$ and $B'$ are of the form $A'=A_{0}'E_A'$ and $B'=B_1'E_B'$, where $A_0'$ is a register of  $\abs{\{z_i\mid z_i=0\}}$ qubits: $A_0'=\otimes_{i:z_i=0}A'_{0_i}$ , i.e. the number of qubits that have to be sent to $V_0$ (same as $\otimes_{i:z_i=0}P^{z_i}$, as described above), similarly, $B_1'=\otimes_{i:z_i=1}B'_{1_i}$ is a register of  $\abs{\{z_i\mid z_i=1\}}$ qubits (same as $\otimes_{i:z_i=1}P^{z_i}$, as described above), and $E_A'$ and $E_B'$  are auxiliary systems. Alice and Bob send the registers $A_0'$ and $B_1'$ to their closest verifier, respectively. 
\end{enumerate}

For a schematic representation of the general attack to \routingfparallel, see \cref{fig:attack-parallel_repBB84} but replacing $\{A^{xy}\}_a$ and $\{B^{xy}\}_b$ by $K^{xy}$ and $L^{xy}$, and the srtaight arrows comming out of the attackers by onbdulated lines, representing $A_0'$ and $B_0'$, respectively.  The tuple $T=\{\ket\varphi,U^x,V^y,L^{xy},K^{xy}\}_{x,y}$ will be called a $q$-qubit strategy for \routingfparallel. Then, the probability that Alice and Bob perform a succesful attack, provided the strategy $T$, which we denote by $\omega_T$, is given by
\begin{equation} \label{eq:w_T}
\begin{split}
    \omega_T(\text{\routingfparallel})&=\frac{1}{2^{2n}}\sum_{\substack{x,y\in\{0,1\}^n\\a:w_H(a)\leq\gamma m}}\tr{N^{f(x,y)}_a\ptr{E_A'E_B'}{( K^{xy}_{A'}\otimes L^{xy}_{B'})\ketbra{\psi_{xy}}{\psi_{xy}}_{VA'B'}(K^{xy}_{A'}\otimes L^{xy}_{B'})^{\dagger}}}.
\end{split}
\end{equation}
Note that \eqref{eq:w_T} can  be equivalently written as 
\begin{equation} \label{eq:w_T no ptr}
\begin{split}
    &\omega_T(\text{\routingfparallel})\\&=\frac{1}{2^n}\sum_{\substack{x,y\in\{0,1\}^n\\a:w_H(a)\leq\gamma m}}\tr{ \left(
    N_a^{f(x,y)}\bigotimes_{i}^{m}\mathbb I_{P_{1-f(x,y)_i}}\otimes \mathbb I_{E_A'E_B'}\right)(\mathbb I_{V}\otimes K^{xy}\otimes L^{xy}\ketbra{\psi_{xy}}{\psi_{xy}}(\mathbb I_{V}\otimes K^{xy}\otimes L^{xy})^{\dagger}}.
\end{split}
\end{equation}
The optimal attack probability is given by
\begin{equation}
    \omega^*(\text{\routingfparallel}):=\sup_{T}\omega_T(\text{\routingfparallel}),
\end{equation}
where the supremum is taking over all possible strategies $T$. As mentioned above, the existence of a generic attack for all QPV protocols \cite{Beigi_2011,Buhrman_2014} implies that $\omega^*(\text{\routingfparallel})$ can be made arbitrarily close to 1. However, the best known attack requires an exponential amount of pre-shared entanglement. Therefore, we will study the optimal winning probability under restricted strategies 
$T$, specifically imposing a constraint on the number of pre-shared qubits $q$ that Alice and Bob hold in step 2 of the general attack. 
Throughout this section, we adopt the following notation to enhance readability:
\begin{enumerate}
    \item we  omit (\text{\routingfparallel}) in $\omega_{T}(\text{\routingfparallel})$, and its variants (see below), and
\item given a strategy $T=\{\ket\varphi,U^x,V^y,K^{xy},L^{xy}\}_{x,y}$,
   \begin{equation}
       \Gamma^{xy}_{LK}:=\sum_{a:w_H(a)\leq\gamma m}\Big(\mathbb I_{V}\otimes K^{xy}\otimes L^{xy}\Big)^{\dagger} \Big(
    N_a^{f(x,y)}\bigotimes_{i}^{m}\mathbb I_{P_{1-f(x,y)_i}}\otimes \mathbb I_{E_A'E_B'}\Big)\Big(\mathbb I_{V}\otimes K^{xy}\otimes L^{xy}\Big),
   \end{equation}
\end{enumerate}
in this way, we have
\begin{equation}
    \omega_T=
    \frac{1}{2^{2n}}\sum_{x,y\in\{0,1\}^n}\tr{\Gamma^{xy}_{KL}\ketbra{\psi_{xy}}{\psi_{xy}}}.
\end{equation}

The key part of this section is \cref{lemma:w_FS_routing}, which is an adapted version of Proposition 7 in \cite{escolàfarràs2024quantumcloninggameapplications}. This will allow us to use the same techniques as in \cref{sec:parallel_f-BB84} to prove security for \routingfparallel. In \cite{escolàfarràs2024quantumcloninggameapplications}, the security of the $m$-fold parallel repetition of \routing (\routingparallel{m}) for the error free case was analyzed in the No-PE model, and the authors showed that the protocol has exponentially small (in the quantum information $m$) soundness, provided that the quantum information travels at the speed of light. Similarly as in \cref{sec:parallel_f-BB84}, consider the \emph{fixed initial-state} (FIS) attack model, which we define as the attack model where step 2. in the general attack in step 2. is constrained by imposing $\ket{\psi_{xy}}\rightarrow \ket{\psi}$ for all $x,y\in\{0,1\}^n$, i.e. strategies of the form $T_{\text{FIS}}=\{\ket\varphi,U^x=\mathbb{I},V^y=\mathbb{I},K^{xy},L^{xy}\}_{x,y}$. Then, the same reduction to a quantum clonning game as in \cite{escolàfarràs2024quantumcloninggameapplications} to show security of \routingparallel{m} holds for \routingfparallel. Not surprisingly, the reduction can be done to strategies $T_{\text{FS}}$ where $K^{xy}$ and $L^{xy}$ only depend on $z=f(x,y)$ instead of $x$ and $y$, i.e. $K^z$ and $L^z$ see proof of Lemma~\ref{lemma:w_FS_routing}.

\begin{lemma}\label{lemma:w_FS_routing}(Adapted version of Proposition 7 in \cite{escolàfarràs2024quantumcloninggameapplications}). For every function $f$ such that reproduces a uniform distribution over $z\in\{0,1\}^m$, the following holds for the error-free case ($\gamma=0$)
\begin{equation}\omega^*_\text{FS}:=\sup_{T_\text{FS}}\omega_{T_\text{FS}}\leq (\mu_0)^m. 
\end{equation} 
\end{lemma}
Recall that $\mu_\gamma$ is defined in \eqref{eq:def_mu_gamma}.
\begin{proof} 
From \eqref{eq:w_T}, we have that for $T_{\text{FIS}}=\{\ket\varphi,U^x=\mathbb{I},V^y=\mathbb{I},K^{xy},L^{xy}\}_{x,y}$,
\begin{equation}\label{eq:T_FIS_error_free}
    \begin{split}
       \omega_{T_{\text{FIS}}}&=\frac{1}{2^{2n}}\sum_{x,y\in\{0,1\}^n}\tr{N_0^{f(x,y)}\ptr{E_A'E_B'}{( K^{xy}\otimes L^{xy})\ketbra{\psi}{\psi}( K^{xy}\otimes L^{xy})^{\dagger}}}\\
       &=\sum_{z}\frac{q_f(z)}{n_z}\sum_{x, y : f(x,y) = z }\tr{N_0^{f(x,y)}\ptr{E_A'E_B'}{( K^{xy}\otimes L^{xy})\ketbra{\psi}{\psi}( K^{xy}\otimes L^{xy})^{\dagger}}},\\
       &\leq\sum_{z}\frac{q_f(z)}{n_z} n_z \max_{x,y:f(x,y)=z}\tr{N_0^{z}\ptr{E_A'E_B'}{( K^{xy}\otimes L^{xy})\ketbra{\psi}{\psi}( K^{xy}\otimes L^{xy})^{\dagger}}}.
\end{split}
\end{equation}
Then, denoting by $L^z$ and $K^z$ the corresponding $L^{xy}$ and $ K^{xy}$ (recall that these $x$ and $y$ are such that $f(x,y)=z$) that attain the maximum in the last inequality, we have that 
\begin{equation}\label{eqw:T_FS}
    \omega_{T_{\text{FS}}} \leq \frac{1}{2^m}\sum_z \tr{N^{z}_0\ptr{E_A'E_B'}{( L^{z}\otimes K^{z})\ketbra{\psi}{\psi}( L^{z}\otimes K^{z})^{\dagger}}} .
\end{equation}
In \cite{escolàfarràs2024quantumcloninggameapplications} (Theorem 4), it is proven that the right-hand-side of \eqref{eqw:T_FS} is upper bounded by $\left(\mu_0\right)^m$. 
\end{proof}

Here, we show an upper bound for $\omega^*_{FS}$ when introducing the error parameter $\gamma\in[0,\frac{1}{2})$. 
\begin{lemma}\label{lemma:w_FS_routing_with_error}(Error-robust version of \cref{lemma:w_FS_routing_with_error}). For every function $f$ such that reproduces a uniform distribution over $z\in\{0,1\}^m$, the following holds for an error parameter $\gamma\in[0,\frac{1}{2})$:
\begin{equation}\omega^*_\text{FS}:=\sup_{T_\text{FS}}\omega_{T_\text{FS}}\leq (\mu_\gamma)^m. 
\end{equation} 
\end{lemma}
Recall that $\mu_\gamma$ is defined in \eqref{eq:def_mu_gamma}. The proof of \cref{lemma:w_FS_routing_with_error}, see Appendix A, consists of a modification of the proof of Proposition 7 in \cite{escolàfarràs2024quantumcloninggameapplications}), inspired by the proof of Theorem~4 in \cite{TomamichelMonogamyGame2013}.

We will define the counterparts of \cref{def:q-beta-strategy}, \cref{def:good_state_for_z}, \cref{def:rounding}, and \cref{def:approximation}. The intuition behind them is the same as in \cref{sec:parallel_f-BB84}, and these concepts will be used to show analogous results to\cref{sec:parallel_f-BB84} in the proof of the main result in this section: security for parallel repetition, stated in  \cref{thm:main_routing_parallel_soundess}.

\begin{definition} Let $\omega_0\in(0,1]$. A $q$-qubit strategy $T$ for \routingfparallel~is a $(\omega_0,q,\beta\cdot2^{2n})$-strategy for \routingfparallel~if there exists a set $\mathcal{B}\subseteq\{0,1\}^{2n}$ with $\abs{\mathcal{B}}\geq \beta\cdot2^{2n}$  such that 
\begin{equation}
    \tr{\Gamma^{xy}_{KL}\ketbra{\psi_{xy}}{\psi_{xy}}}\geq \omega_0, \text{ }\text{ } \forall(x,y)\in\mathcal{B}.
\end{equation}
\end{definition}
\begin{definition}
   Let $\varepsilon,\Delta>0$. We say that a state $\ket{\psi}_{VA'B'}$ is $\Delta-$\emph{good to attack} $z\in\{0,1\}^m$ for \routingfparallel~if there exists \text{unitaries} $K^z$ and $L^z$ acting on $A'$ and $B'$, respectively, such that the probability that the verifiers \emph{accept} on input $z$ (the left hand side of the following inequality) is such that
    \begin{equation}
       \tr{N^{z}\ptr{E_A'E_B'}{( K^{z}\otimes L^{z})\ketbra{\psi}{\psi}(K^{z}\otimes L^{z})^{\dagger}}}\geq (\mu_\gamma+\Delta)^m\left(1+3\sqrt{3\ln{(2/\varepsilon)}}2^{-n+m/2}\right).
    \end{equation}
\end{definition}

\begin{definition}
   Let $\omega_0\in(0,1]$, $\Delta>0$,  $s=1-\log\frac{\lambda_0+\Delta}{\lambda_0}$ and $k_1,k_2,k_3\in\mathbb N$.  A function 
    \begin{equation}
        g:\{0,1\}^{k_1}\times \{0,1\}^{k_2}\times \{0,1\}^{k_3}\rightarrow 
        \mathcal P_{\leq s}(\{0,1\}^m)
    \end{equation}
    is a $(\omega_0,q)$-set-valued classical rounding for \routingfparallel~ of sizes $k_1,k_2,k_3$ if for all functions $f\in\mathcal F^*_\varepsilon$, all $\ell\in\{1,\ldots,2^{2n}\},$ for all $(\omega_0,q,\ell)-$strategies for \routingfparallel, there exist functions ${f_A:\{0,1\}^n\rightarrow\{0,1\}^{k_1}}$, $f_B:\{0,1\}^n\rightarrow\{0,1\}^{k_2}$ and $\lambda\in\{0,1\}^{k_3}$ such that, on at least $\ell$ pairs $(x,y)$,
    \begin{equation}
        f(x,y)\in g(f_A(x),f_B(y),\lambda).
    \end{equation}
\end{definition}

\begin{definition}Let $\delta\in(0,1)$. A $\delta-$\emph{approximation} of a strategy $T=\{\ket\varphi,U^x,V^y,K^{xy},L^{x,y}\}_{x,y}$ is the tuple $T_\delta=\{\ket{\varphi_\delta},U^x_\delta,V_\delta^y,K^{xy},L^{xy}\}_{x,y}$, where $\ket{\varphi_\delta}$, $U^x_\delta$ and $V_\delta^y$ are such that, for every $x,y\in\{0,1\}^n$,
\begin{equation}
    \norm{\ket{\varphi}-\ket{\varphi_\delta}}_2\leq \delta, \text{ }\norm{U^x-U^x_\delta}_\infty\leq \delta,  \text{ and } \norm{V^y-V_\delta^y}_\infty\leq \delta.  
\end{equation}
\end{definition}

Now we state our main result showing security of \routingfparallel. In its proof, we will show that the lemmas in \cref{sec:parallel_f-BB84} have an analogous version for \routingfparallel. 

\begin{theorem}\label{thm:main_routing_parallel_soundess} Let $n>m$, $\varepsilon\leq 2^{-m-1}$ and $\Delta>0$. For every $c<1$, with probability at least $1-2^{-m2^{n-c m \log(\frac{1}{\mu_\gamma+\Delta})}}$, a uniformly random $f\in \mathcal F^*_\varepsilon$ will be such that, if the number of qubits $q$ that the attackers pre-share to attack \routingfparallel~is such that
   \begin{equation}\label{eq:soundess_theorem_routing}
        2q< n-c m\log\left(\frac{1}{\mu_\gamma+\Delta}\right)+\log\frac{(1-(\mu_\gamma+\Delta)^{1-c})\log\left(\frac{\mu_\gamma+\Delta}{\mu_\gamma}\right)}{8\log(1/\Delta)},
    \end{equation}
     then, the probability that the verifiers \emph{accept} is at most 
     \begin{equation}\label{eq:w_theorem_routing}
      \left((\mu_\gamma+\Delta)^{c}\right)^m(1+3\sqrt{3\ln(2/\varepsilon)}2^{-n+m/2})+7\cdot 3\Delta^m.
    \end{equation}
\end{theorem}

Notice that the bound in \cref{thm:main_routing_parallel_soundess} exponentially decays in $m$ if $\mu_\gamma+\Delta<1$. Moreover, since, by hypothesis $\varepsilon\leq 2^{-m-1}$, in particular we have that, under the conditions of \cref{thm:main_bb84_parallel_soundess}, any $q$-qubit strategy $T$ for \routingfparallel~is such that 
\begin{equation}
    \omega_T\leq\left((\mu_\gamma+\Delta)^{c}\right)^m(1+3\sqrt{3m\ln(2)}2^{-n+m/2})+7\cdot 3\Delta^m.
\end{equation}
\cref{thm:main_routing_parallel_soundess} leaves freedom to pick the values $\Delta$ and $c$. If one wants a lower upper bound on the soundness, these should be picked small and big, respectively. By picking $\Delta$ small enough, e.g, $\Delta=10^{-5}$, the term $\lambda_\gamma+\Delta$ is strictly smaller than 1 for an error $\gamma$ up to roughly $3.0\%$ and we have that up to that error, the upper bound on the soundness in \cref{thm:main_routing_parallel_soundess} will decay exponentially. Notice that the asymptotic behavior of the upper bound on the soundness behaves as 
\begin{equation}
    \left((\mu_\gamma+\Delta)^{c}\right)^m.
\end{equation}

\begin{proof}
Let $\ket{\psi}_{VA'B'}$ be the state that Alice and Bob use in a strategy $T_{\text{FIS}}$, and consider
\begin{equation}
    \omega_{\psi}^*:=\max_{\{L^{xy},K^{xy}\}}\frac{1}{2^{2n}}\sum_{x,y\in\{0,1\}^n}\tr{N^{f(x,y)}
    \ptr{E_A'E_B'}{( K^{xy}\otimes L^{xy})\ketbra{\psi}{\psi}( K^{xy}\otimes L^{xy})^{\dagger}}}.
\end{equation}
Analogously to \cref{lem:bound w_psi for f in F}, we have that for every $f\in\mathcal F_{\varepsilon}$  the following bound holds for \routingfparallel: for every quantum state $\ket{\psi}_{VA'B'}$, with arbitrary dimensional registers $A'$ and $B'$,
\begin{equation}
    \omega_\psi^*\leq \left(\mu_\gamma\right)^m\big(1+\sqrt{3\ln{(2/\varepsilon)}}2^{-n+m/2}\big).
    \end{equation}
Then, in the same as shown in \cref{lem:number_z}, we have that for 
 $\Delta>0$, for every $f\in\mathcal F^*_\varepsilon$, any quantum state $\ket{\psi}_{VA'B'}$ can be $\Delta-$good for \routingfparallel~on at most a fraction of all the possible $z\in\{0,1\}^m$ given by 
    \begin{equation}
        \left(\frac{\mu_\gamma}{\mu_\gamma+\Delta}\right)^m.
    \end{equation}

On the other hand, \cref{lem:delta_approx is good} has its counterpart for the routing version: 
Let \\$T=\{\ket\varphi,U^x,V^y,K^{xy},L^{xy}\}_{x,y}$ be a $q-$qubit strategy for \routingfparallel. Then, every $\delta-$\emph{approximation} of $T$, fulfills the following inequality for all $(x,y)$:
\begin{equation}
    \tr{\Gamma^{xy}_{KL}\ketbra{\psi^\delta_{xy}}{\psi^\delta_{xy}}}\geq \tr{\Gamma^{xy}_{KL}\ketbra{\psi_{xy}}{\psi_{xy}}}-7\delta.
\end{equation}

We then use this to construct a $(\omega_0,q)$-set-valued classical rounding for \routingfparallel~ of sizes $k_1,k_2\leq\log_2\left(\frac{1}{\Delta}\right)m2^{2q+1}, \text{ and } k_3\leq\log_2\left(\frac{1}{\Delta}\right)m2^{2q+m+1}$, with $\omega_0\geq (\lambda_\gamma+\Delta)^{m}(1+3\sqrt{3\ln(2/\varepsilon)}2^{-n+m/2})+7\cdot3\Delta^m$ for $\Delta>0$, in the same way as in \cref{lem:existence_rounding} by replacing \eqref{eq;def_g} by 
\begin{equation}\label{eq;def_g_routing}
\begin{split}
     &g(f_A(x),f_B(y),\lambda):=\\&\{z\mid \exists K^z,L^z \text{ with } 
    \sum_{a:w_H(a)\leq \gamma m}\tr{N^{z}_a\ptr{E_A'E_B'}{( K^{z}\otimes L^{z})\ketbra{\psi_{xy}^\delta}{\psi_{xy}^\delta}(K^{z}\otimes L^{z})^{\dagger}}}
    \geq \omega_0-7\delta \}.
\end{split}
\end{equation}
Then, \cref{lem:q_bounded_in_rounding} holds for \routingfparallel. In addition, the analogous version of \cref{lem:at leas b pairs are good} applies to \routingfparallel. Then, \cref{thm:main_routing_parallel_soundess} is proved in the same way as \cref{thm:main_bb84_parallel_soundess}.

\end{proof}

\subsection{Improved error-tolerance for \routingf}

In \cite{bluhm2022single}, it was shown that \routingf~ is secure for attackers who pre-share a linear amount (in $n$) of qubits as long as the error remains below $4\%$. Here, by considering the case $m=1$ in \QPVBBfparallel, which corresponds to \QPVBBf, we show that the protocol can tolerate an error almost up to $25\%$, presenting an order-of-magnitude improvement in error tolerance.

For the case of $m=1$, from \cite{escolàfarràs2024quantumcloninggameapplications}~we have the tight result $\omega_{T_{FS}}\leq\frac{3}{4}$. Then, with the same analysis, we can make the upper bound in \cref{thm:main_routing_parallel_soundess} tighter (we used the non-tight result $\omega_{T_{FS}}\leq\mu_0$ for $m=1)$. Then, by picking, e.g. $\Delta=10^{-5}$ and $c=0.999$, then we have the following corollary:

\begin{corollary}
\label{coro:1_round_routing} Let $n,\in\mathbb N$, with $n>m$ and $n\geq 36$, and $\varepsilon\leq 2^{-m-1}$. Then, with probability at least $1-2^{-2^{n-c  \log(\frac{1}{3/4+\Delta})}}$, a uniformly random $f\in \mathcal F^*_\varepsilon$ will be such that, if 
   \begin{equation}\label{eq:soundess_1_round_routing}
        q< \frac{1}{2}n+\frac{1}{2}\log\frac{(3/4+\Delta)^c(1-(3/4+\Delta)^{1-c})\log\left(\frac{3/4+\Delta}{3/4}\right)}{8\log(1/\Delta)}\simeq\frac{1}{2}n-17.449,
    \end{equation}
     any $q$-qubit strategy $T$ for \routingfparallel~is such that 
     \begin{equation}\label{eq:w_1_round_routing}
         \omega_T\leq\left(\frac{3}{4}+\Delta\right)^{c}(1+3\sqrt{6\ln(2)}2^{-n})+7\cdot 3\Delta\simeq 0.750226(1+3\sqrt{6\ln(2)}2^{-n})+0.00021.
    \end{equation}
\end{corollary}

Thus, the upper bound in \eqref{eq:w_1_round_routing} converges exponentially in $n$ to 
\begin{equation}
   0.750436\ldots
\end{equation}
Notably, the attack described in \cite{escolàfarràs2024quantumcloninggameapplications}, which uses no pre-shared entanglement, achieves a success probability of $\frac{3}{4}=0.75$, showing that our bound is essentially tight. This implies that even if Alice and Bob share a linear amount $q=O(n)$ of pre-shared qubits, they cannot outperform an attack that relies on no pre-shared entanglement.

\section{Discussion}

We have seen that parallel repetition of \QPVBBfparallel~and \routingfparallel~hold and that this implies that a single interaction with a prover suffices to have secure quantum position verification with the security relying on the size of classical information. On the other hand, in practice, a sizable fraction of photons is lost in transmission, and in \cite{allerstorfer2023makingexistingquantumposition}, it was shown a modification of the structure of \QPVBBf~that makes the transmission loss irrelevant for security, while inheriting the properties of the lossless version. We leave as an open question whether this modification makes \QPVBBfparallel~loss tolerant in a robust way. 

\bibliographystyle{alphaurl}
\bibliography{biblio.bib}
\newpage
\begin{appendices}
    \section{Proof of \cref{lemma:w_FS_routing_with_error}}\label{appendix}

Based on a modification in the \cite{escolàfarràs2024quantumcloninggameapplications}, we show the proof of \cref{lemma:w_FS_routing_with_error}. For that, we need the following definition and lemmas. 
\begin{definition} Let $N\in\mathbb N$. Two permutations $\pi,\pi':[N]\rightarrow[N]$ are said to be \emph{orthogonal} if $\pi(i)\neq\pi'(i)$ for all $i\in[N]$. 
\end{definition}

\begin{lemma} \label{lem:sum_projectors}\emph{(Lemma 2 in \cite{TomamichelMonogamyGame2013})}
    Let $\Pi^1,\ldots,\Pi^N$ be projectors acting on a Hilbert space $\mathcal H$. Let $\{\pi_k\}_{k\in[N]}$ be a set of mutually orthogonal permutations. Then, 
    \begin{equation}
        \bigg\|\sum_{i\in[N]}\Pi^i \bigg\|\leq \sum_{k\in[N]}\max_{i\in[N]}\big\| \Pi^i\Pi^{\pi_k(i)}\big\|.
    \end{equation}
\end{lemma}

\begin{remark}
    There always exist a set of $N$ permutations of $[N]$ that are mutually orthogonal, an example is the $N$ cyclic shifts.  
\end{remark}

\begin{lemma}\label{lem:norm_product}(Lemma 1 in \cite{TomamichelMonogamyGame2013})
    Let $A,B,L\in\mathcal{B}(\mathcal{H})$ such that $AA^\dagger\succeq B^\dagger B$. Then it holds that $\norm{AL}\geq\norm{BL}$. 
\end{lemma}

Now, we are in position to prove \cref{lemma:w_FS_routing_with_error}. Let $T_{\text{FIS}}=\{\ket\varphi,U^x=\mathbb{I},V^y=\mathbb{I},K^{xy},L^{xy}\}_{x,y}$, with $\ket\psi=\ket\varphi$, then

\begin{equation}
    \begin{split}
       \omega_{T_{\text{FIS}}}&=\frac{1}{2^{2n}}\sum_{\substack{x,y\in\{0,1\}^n\\a:w_H(a)\leq\gamma m}}\tr{N_a^{f(x,y)}\ptr{E_A'E_B'}{( K^{xy}\otimes L^{xy})\ketbra{\psi}{\psi}( K^{xy}\otimes L^{xy})^{\dagger}}}\\
       &=\sum_{\substack{z\\a:w_H(a)\leq\gamma m}}\frac{q_f(z)}{n_z}\sum_{x, y : f(x,y) = z }\tr{N_a^{f(x,y)}\ptr{E_A'E_B'}{( K^{xy}\otimes L^{xy})\ketbra{\psi}{\psi}( K^{xy}\otimes L^{xy})^{\dagger}}},\\
       &\leq\sum_{\substack{z\\a:w_H(a)\leq\gamma m}}\frac{q_f(z)}{n_z} n_z \max_{x,y:f(x,y)=z}\tr{N_a^{z}\ptr{E_A'E_B'}{( K^{xy}\otimes L^{xy})\ketbra{\psi}{\psi}( K^{xy}\otimes L^{xy})^{\dagger}}}
\end{split}
\end{equation}
Then,  by hypothesis $q_f(z)=\frac{1}{2^m}$, and, denoting by $K^z$ and $L^z$ the corresponding $K^{xy}$ and $ L^{xy}$ (recall that these $x$ and $y$ are such that $f(x,y)=z$) that attain the maximum in the last inequality, we have that 
\begin{equation}\label{eqw:T_FS2}
    \omega_{T_{\text{FS}}} \leq \frac{1}{2^m}\sum_{\substack{z\\a:w_H(a)\leq\gamma m}}\tr{N_a^{z}\ptr{E_A'E_B'}{( K^{z}\otimes L^{z})\ketbra{\psi}{\psi}( K^{z}\otimes L^{z})^{\dagger}}}.
\end{equation}
As shown in \eqref{eq:w_T no ptr}, \eqref{eqw:T_FS2} is equivalent to 
\begin{equation}
    \omega_{T_{\text{FS}}} \leq \frac{1}{2^m}\sum_{\substack{z\\a:w_H(a)\leq\gamma m}}\tr{(N_a^{z}\otimes \mathbb I)( K^{z}\otimes L^{z})\ketbra{\psi}{\psi}( K^{z}\otimes L^{z})^{\dagger}},
\end{equation}
where the identity $\mathbb I$ applies to all the remaining registers, see \eqref{eq:w_T no ptr} for the explicit registers. Then, 
\begin{equation}
\begin{split}
    \omega_{T_{\text{FS}}} &\leq \frac{1}{2^m}\sum_{\substack{z\\a:w_H(a)\leq\gamma m}}\tr{( K^{z}\otimes L^{z})^{\dagger}(N_a^{z}\otimes \mathbb I)( K^{z}\otimes L^{z})\ketbra{\psi}{\psi}}\\
    &\leq \frac{1}{2^m}\norm{\sum_{\substack{z\\a:w_H(a)\leq\gamma m}}( K^{z}\otimes L^{z})^{\dagger}(N_a^{z}\otimes \mathbb I)( K^{z}\otimes L^{z})}.
    \end{split}
\end{equation}
Let
\begin{equation}
    \Tilde N^{z}_a:=( K^{z}\otimes L^{z})^{\dagger}(N_a^{z}\otimes \mathbb I)( K^{z}\otimes L^{z}),
\end{equation}
then, 
\begin{equation}
     \omega_{T_{\text{FS}}} \leq \frac{1}{2^m}\norm{\sum_{\substack{z\\a:w_H(a)\leq\gamma m}}\Tilde N^z_a}\leq\frac{1}{2^m}\sum_{a:w_H(a)\leq\gamma m}\sum_{k\in[2^m]}\max_{z,z'}\norm{\Tilde{N}^z_a\Tilde N^{z'}_a},
\end{equation}
where we used Lemma~\ref{lem:sum_projectors}, and $z'=\pi_k(z)$, for $\{\pi_k\}_k$ being a set of mutually orthogonal permutations. Fix $z$ and $z'$, and let $\mathcal T$ be the set of indices where $z$ and $z'$ differ, i.e.~$\mathcal{T}=\{i\mid z_i\neq z'_i\}$, and let $t=\abs{\mathcal{T}}$. Let $\mathcal{T}_A=\{i\in \mathcal{T}\mid z_i=0 \}$, denote  $t_A:=\abs{\mathcal{T}_A}$, and, without loss of generality, assume $t_A\geq t/2$. Let $\mathcal{T}_A^0=\{i\in \mathcal{T}_A\mid a_i=0\}$, and $t_A^0:=\abs{\mathcal{T}_A^0}$, then we have that
\begin{equation}
    \begin{split}
        &\Tilde N^z_a\preceq \Tilde N^z_{a_A}:=\\&(\mathbb I_{V}\otimes K^{z\dagger}_{A_0'E_A'}\otimes L^{z\dagger}_{B_1'E_B'})\left(\bigg(\bigotimes_{i\in \mathcal{T}_A^0}\ketbra{\Phi^+}{\Phi^+}_{V_0^iA'_{0_i}}\otimes\mathbb I_{B'_{1_{i}}}\bigg)\otimes\bigg(\bigotimes_{i\in [m]\setminus\mathcal{T}_A^0}\mathbb I_{V_0^iP^{z_i}P^{1-z_i}}\bigg)\otimes \mathbb I_{E_A'E_B'}\right)(\mathbb I_{V}\otimes K^z_{A_0'E_A'}\otimes L^z_{B_1'E_B'})\\
        &=(\mathbb I_{V}\otimes K^{z\dagger}_{A_0'E_A'}\otimes L^{z\dagger}_{B_1'E_B'})\left(\bigg(\bigotimes_{i\in \mathcal{T}_A^0}\ketbra{\Phi^+}{\Phi^+}_{V_0^iA'_{0_i}}\bigotimes_{i\in [m]\setminus\mathcal{T}_A^0}\mathbb{I}_{V_0^iA'_{0_i}E_A'}\right)\otimes \mathbb I_{B'_1E_B'} \bigg)(\mathbb I_{V}\otimes K^z_{A_0'E_A'}\otimes L^z_{B_1'E_B'})\\
        &=(\mathbb I_{V}\otimes K^{z\dagger}_{A_0'E_A'}\otimes L^{z'\dagger}_{B_1'E_B'})\left(\bigg(\bigotimes_{i\in \mathcal{T}_A^0}\ketbra{\Phi^+}{\Phi^+}_{V_0^iA'_{0_i}}\bigotimes_{i\in [m]\setminus\mathcal{T}_A^0}\mathbb{I}_{V_0^iA'_{0_i}E_A'}\right)\otimes \mathbb I_{B'E_B'} \bigg)(\mathbb I_{V}\otimes K^z_{A_0'E_A'}\otimes L^{z'}_{B_1'E_B'})
            \end{split}
\end{equation}
where in the last equality we used that $L^{z \dagger}_{B_1'E_B'}L^{z \dagger}_{B_1'E_B'}=\mathbb I_{B'_1E_B'}=L^{z' \dagger}_{B_1'E_B'}L^{z'}_{B_1'E_B'}$. Similarly, 
\begin{equation}
    \begin{split}
        &\Tilde N^{z'}_a\preceq \Tilde N^{z'}_{a_B}:=\\&(\mathbb I_{V}\otimes K^{z\dagger}_{A_0'E_A'}\otimes L^{z\dagger}_{B_1'E_B'})\left(\bigg(\bigotimes_{i\in \mathcal{T}_A^0}\ketbra{\Phi^+}{\Phi^+}_{V_0^iB'_{1_i}}\otimes\mathbb I_{A'_{0_{i}}}\bigg)\otimes\bigg(\bigotimes_{i\in [m]\setminus\mathcal{T}_A^0}\mathbb I_{V_0^iP^{z_i}P^{1-z_i}}\bigg)\otimes \mathbb I_{E_A'E_B'}\right)(\mathbb I_{V}\otimes K^z_{A_0'E_A'}\otimes L^z_{B_1'E_B'})\\
        &=(\mathbb I_{V}\otimes K^{z\dagger}_{A_0'E_A'}\otimes L^{z\dagger}_{B_1'E_B'})\left(\bigg(\bigotimes_{i\in \mathcal{T}_A^0}\ketbra{\Phi^+}{\Phi^+}_{V_0^iB'_{1_i}}\bigotimes_{i\in [m]\setminus\mathcal{T}_A^0}\mathbb{I}_{V_0^iA'_{0_i}E_A'}\right)\otimes \mathbb I_{B'_1E_B'} \bigg)(\mathbb I_{V}\otimes K^z_{A_0'E_A'}\otimes L^z_{B_1'E_B'})\\
        &=(\mathbb I_{V}\otimes K^{z'\dagger}_{A_0'E_A'}\otimes L^{z\dagger}_{B_1'E_B'})\left(\bigg(\bigotimes_{i\in \mathcal{T}_A^0}\ketbra{\Phi^+}{\Phi^+}_{V_0^iB'_{1_i}}\bigotimes_{i\in [m]\setminus\mathcal{T}_A^0}\mathbb{I}_{V_0^iA'_{0_i}E_A'}\right)\otimes \mathbb I_{B'_1E_B'} \bigg)(\mathbb I_{V}\otimes K^{z'}_{A_0'E_A'}\otimes L^z_{B_1'E_B'})
            \end{split}
\end{equation}
By Lemma~\ref{lem:norm_product}, 
\begin{equation}\label{eq:MxMx<=...}
    \norm{\Tilde N^z_a \Tilde N^{z'}_a}\leq\norm{\Tilde N^{z}_{a_A}\Tilde N^{z'}_{a_B}},
\end{equation}
then,

\begin{equation}
\begin{split}
   \Tilde N^{z}_{a_A}\Tilde N^{z'}_{a_B}&=(\mathbb I_{V}\otimes K^{z\dagger}_{A_0'E_A'}\otimes L^{z'\dagger}_{B_1'E_B'})\left(\bigg(\bigotimes_{i\in \mathcal{T}_A^0}\ketbra{\Phi^+}{\Phi^+}_{V_0^iA'_{0_i}}\bigotimes_{i\in [m]\setminus\mathcal{T}_A^0}\mathbb{I}_{V_0^iA'_{0_i}E_A'}\right)\otimes \mathbb I_{B'E_B'} \bigg)(\mathbb I_{V}\otimes K^z_{A_0'E_A'}\otimes L^{z'}_{B_1'E_B'})\\
    &\cdot (\mathbb I_{V}\otimes K^{z'\dagger}_{A_0'E_A'}\otimes L^{z\dagger}_{B_1'E_B'})\left(\bigg(\bigotimes_{i\in \mathcal{T}_A^0}\ketbra{\Phi^+}{\Phi^+}_{V_0^iB'_{1_i}}\bigotimes_{i\in [m]\setminus\mathcal{T}_A^0}\mathbb{I}_{V_0^iA'_{0_i}E_A'}\right)\otimes \mathbb I_{B'_1E_B'} \bigg)(\mathbb I_{V}\otimes K^{z'}_{A_0'E_A'}\otimes L^z_{B_1'E_B'})
    \end{split}
\end{equation}
We have that $(\mathbb I_{V}\otimes K^z_{A_0'E_A'}\otimes L^{z'}_{B_1'E_B'}) (\mathbb I_{V}\otimes K^{z'\dagger}_{A_0'E_A'}\otimes L^{z\dagger}_{B_1'E_B'})=\mathbb{I}_{VA'B'}$, and, since the Schatten $\infty$-norm is unitarily invariant, 

\begin{equation}\label{eq:normMxAMxB}
\begin{split}
     \norm{\Tilde N^{z}_{a_A}\Tilde N^{z'}_{a_B}}&=\norm{\left(\bigotimes_{i\in \mathcal{T}_A^0}\ketbra{\Phi^+}{\Phi^+}_{V_0^iA'_{0_i}}\bigotimes_{i\in [m]\setminus\mathcal{T}_A^0}\mathbb{I}_{V_0^iA'_{0_i}E'_A}\otimes \mathbb I_{B_1'E_B'} \right)\left(\bigotimes_{i\in \mathcal{T}_A^0}\ketbra{\Phi^+}{\Phi^+}_{V_0^iB'_{1_ i}}\bigotimes_{i\in [m]\setminus\mathcal{T}_A^0}\mathbb{I}_{V_0^iB'_{1_i}E_B'}\otimes \mathbb I_{A_0'E_A'} \right)}\\
     &=\norm{\left(\bigotimes_{i\in \mathcal{T}_A^0}\big(\ketbra{\Phi^+}{\Phi^+}_{V_0^iA'_{0_i}}\otimes\mathbb I_{B'_{1_i}}\big)\big(\ketbra{\Phi^+}{\Phi^+}_{V_0^iB'_{1_i}}\otimes\mathbb I_{A'_{0_i}}\big)\right)\bigotimes_{i\in [m]\setminus\mathcal{T}_A^0}\mathbb I_{V_0^iA'_{0_i}B'_{1_i}}\otimes\mathbb{I}_{E_A'E_B'}}
     \\&=\norm{\bigotimes_{i\in \mathcal{T}_A^0}\big(\ketbra{\Phi^+}{\Phi^+}_{V_0^iA'_{0_i}}\otimes\mathbb I_{B'_{1_i}}\big)\big(\ketbra{\Phi^+}{\Phi^+}_{V_0^iB'_{1_i}}\otimes\mathbb I_{A'_{0_i}}\big)}\norm{\bigotimes_{i\in [m]\setminus\mathcal{T}_A^0}\mathbb I_{V_0^iA'_{0_i}B'_{0_i}}\otimes\mathbb{I}_{E_A'E_B'}}\\
     &=\prod_{i\in\mathcal{T}_A^0}\norm{\big(\ketbra{\Phi^+}{\Phi^+}_{V_0^iA'_{0_i}}\otimes\mathbb I_{B'_{1_i}}\big)\big(\ketbra{\Phi^+}{\Phi^+}_{V_0^iB'_{1_i}}\otimes\mathbb I_{A'_{0_i}}\big)}
     \\&=2^{-t_A},
\end{split}
\end{equation}
where we used that, for every $i$,
\begin{equation}
    \norm{\big(\ketbra{\Phi^+}{\Phi^+}_{V_0^iA'_{0_i}}\otimes\mathbb I_{B'_{1_i}}\big)\big(\ketbra{\Phi^+}{\Phi^+}_{V_0^iB'_{1_i}}\otimes\mathbb I_{A'_{0_i}}\big)}=2^{-1}.
\end{equation}

Let  $t_A^1:=\abs{\{i\in \mathcal T_A\mid a_i=1\}}$, then, since in order to accept, $w_H(a)\leq \gamma m$, in particular, we have that $t_A^1\leq \gamma m$. Then, using that $t_A^0=t_A-t_A^1\geq t/2-\gamma m$, where we used that $t_A\geq t/2$. Then, combining \eqref{eq:MxMx<=...} and \eqref{eq:normMxAMxB}, we have that 
\begin{equation}
    \norm{\Tilde N^z_a \Tilde N^{z'}_a}\leq\norm{\Tilde N^{z}_{a_A}\Tilde N^{z'}_{a_B}}\leq 2^{-\frac{t}{2}+\gamma m}
\end{equation}

In order to apply the bound in Lemma~\ref{lem:norm_product}, consider the set of permutations given by $\pi_k(z)=z\oplus k$, where $z, k\in \{0,1\}^{m}$ (they are such that they have the same Hamming distance). There are $\binom{m}{i}$ permutations with Hamming distance $i$. Then, we have
\begin{equation}
    \omega_{T_{\text{FS}}} \leq \frac{1}{2^m}\sum_{a:w_H(a)\leq \gamma m}\sum_{k\in[2^m]}\max_{z,z'}\norm{\Tilde N^z_a \Tilde N^{z'}_a}\leq \frac{1}{2^m}\sum_{a:w_H(a)\leq \gamma m}\sum_{t=0}^m\binom{m}{t}2^{-\frac{t}{2}+\gamma m}=\left(2^{\gamma+h(\gamma)}\Big(\frac{1}{2}+\frac{1}{2\sqrt{2}}\Big)\right)^m,
\end{equation}
where we used that $\sum_{a:w_H(a)\leq \gamma m}\leq 2^{h(\gamma)m}$, for $\gamma\leq 1/2$.
\end{appendices}

\end{document}